\documentclass[a4paper,twocolumn,10pt,accepted=2020-07-24]{quantumarticle}
\pdfoutput=1
\usepackage[utf8]{inputenc}
\usepackage[english]{babel}
\usepackage[T1]{fontenc}
\usepackage{amsmath}
\usepackage{hyperref}
\usepackage{tikz}
\usepackage{lipsum}
\usepackage{graphicx,amsmath,amssymb,color}
\usepackage{float}
\usepackage{graphicx}
\usepackage{dcolumn}
\usepackage{bm}
\usepackage{amsthm}
\usepackage{comment}
\usepackage{booktabs}
\usepackage{IEEEtrantools}
\usepackage{braket}
\usepackage{amsmath}
\usepackage{mathtools}
\usepackage{epigraph}
\usepackage{physics}

\usepackage[numbers,sort&compress]{natbib}
\newtheorem{Proposition}{Proposition} 
\newtheorem{Observation}{Observation} 
\newtheorem{Definition}{Definition}
\newcommand{\be}{\begin{equation}}
\newcommand{\ee}{\end{equation}}
\newcommand{\beq}{\begin{eqnarray}}
\newcommand{\eeq}{\end{eqnarray}}

\newcommand{\ds}{\color{blue}}  
\newcommand{\blk}{\color{black}}

\begin{document}

\title{Quantum prescriptions are ontologically more distinct than they are operationally distinguishable}

\author{Anubhav Chaturvedi}
\email{anubhav.chaturvedi@research.iiit.ac.in}
\affiliation{Institute of Theoretical Physics and Astrophysics, National Quantum Information Centre, Faculty of Mathematics, Physics and Informatics, University of Gda\'{n}sk, 80-952 Gda\'{n}sk, Poland}
\affiliation{International Centre for Theory of Quantum Technologies (ICTQT), University of Gdansk, 80-308 Gda\'nsk, Poland}

\author{Debashis Saha}
\email{saha@cft.edu.pl}
\affiliation{Institute of Theoretical Physics and Astrophysics, National Quantum Information Centre, Faculty of Mathematics, Physics and Informatics, University of Gda\'{n}sk, 80-952 Gda\'{n}sk, Poland}
\affiliation{Center for Theoretical Physics, Polish Academy of Sciences, Aleja Lotnik\'{o}w 32/46, 02-668 Warsaw, Poland}

\maketitle

\begin{abstract}

Based on an intuitive generalization of the Leibniz principle of `the identity of indiscernibles', we introduce a novel ontological notion of classicality, called bounded ontological distinctness. Formulated as a principle, bounded ontological distinctness equates the distinguishability of a set of operational physical entities to the distinctness of their ontological counterparts. Employing three instances of two-dimensional quantum preparations, we demonstrate the violation of bounded ontological distinctness or excess ontological distinctness of quantum preparations, without invoking any additional assumptions. Moreover, our methodology enables the inference of tight lower bounds on the extent of excess ontological distinctness of quantum preparations. Similarly, we demonstrate excess ontological distinctness of quantum transformations, using three two-dimensional unitary transformations.
However, to demonstrate excess ontological distinctness of quantum measurements, an additional assumption such as outcome determinism or bounded ontological distinctness of preparations is required. Moreover, we show that quantum violations of other well-known ontological principles implicate quantum excess ontological distinctness. Finally, to showcase the operational vitality of excess ontological distinctness, we introduce two distinct classes of communication tasks powered by excess ontological distinctness.  
\end{abstract}
\section{Introduction}
Operational physical theories such as quantum theory serve to instruct experiments and make corresponding predictions. However, owing to their rather abstract mathematical formalism, these operational theories are often incapable (on their own) of facilitating deep insights into the nature and structure of reality. The study of ontology along with the ontological framework enables insights into the structure of reality the operational theories posit. Apart from this, the ontological framework also provides a vital ground for formal notions of classicality. These ontological notions of classicality attempt to make precise the sense in which quantum theory departs from classical physics \cite{mazurek2016experimental,spekkens2019ontological}, and are typically philosophically well-substantiated ontological principles which seek to explain certain exclusively operational phenomena by providing them an exclusively ontological basis. For instance, the ontological principle of Bell's local-causality \cite{bell1964einstein,brunner2014bell} attributes a local ontological basis to operational non-signaling correlations. Similarly, Kochen-Specker noncontextuality \cite{kochen1967problem,cabello2008experimentally,cabello2014graph} provides operational non-disturbance of measurements a noncontextual basis. Based directly on the philosophical premise of the Leibniz principle of ``Identity of Indiscernibles", Spekkens' generalized noncontextuality \cite{spekkens2005contextuality} assigns identical ontological counterparts to operationally indistinguishable preparations, measurements, and transformations.

In this work, we propose a novel notion of classicality, ``bounded ontological distinctness". The extent to which physical entities may be distinguished or discerned apart, termed as their distinguishability, plays a crucial role in physical theories, more generally in natural philosophy, as well as in our very world views \cite{bohm2004thought}.
For instance, a principle of analytic philosophy, ``the Identity of Indiscernibles", first explicitly formulated by Wilhelm Gottfried Leibniz is inherent in classical intuition. The principle states that operational indistinguishability implies ontological identity, i.e. if two objects cannot possibly be distinguished, then they are the same (identical) \cite{sep-identity-indiscernible,spekkens2019ontological}. A natural generalization of the Leibniz's principle is the notion that how well a set of objects may be in principle distinguished 
 (i.e. employing all possible measurements) reflects how distinct they actually are. The ontological principle, ``bounded ontological distinctness" relies on this generalization of the Leibniz principle and states that maximal operational distinguishability of physical entities reflects their ontological distinctness. Here, we refer to the extent to which physical entities may be distinguished whilst employing an operational physical theory as their ``distinguishability" as opposed to their ``distinctness", which refers to how ontologically (actually) distinct these entities are. In this work, we formulate the principle of bounded ontological distinctness, and demonstrate how quantum preparations, measurements, and transformations violate this principle, showcasing their ``excess ontological distinctness". 
 
\textit{Requirements for a fundamental notions of classicality:} To be of fundamental significance \cite{spekkensvideo,spekkens2019ontological}, a notion of classicality should be applicable in a wide variety of scenarios so that it covers the multitude of experimental tests posited by quantum theory. Furthermore, to reveal the inherent non-classicality of a large class of quantum systems, such a fundamental notion should impose minimal demands on quantum preparations, transformations, and measurements required to violate it. Similarly, such a notion should only impose minimal, operationally justified constraints on the ontological explanation. Additionally, a fundamental ontological notion of classicality should have as its implications the other well-known notions of classicality, so that quantum violation of the latter implicates the violation of the former \cite{PhysRevA.100.022108}. As the experimental tests of the consequences of these ontological principles allow us to extend the corresponding implications beyond the particular operational theory to Nature itself \cite{mazurek2016experimental}, such an ontological notion of classicality should be robust to experimental imperfections.  Finally, such a notion should be operationally vital so that its quantum violation yields an advantage in computation, communication or information processing tasks. 

\textit{Primary and auxiliary ontological constraints:} All of the aforementioned ontological principles obtain their primary ontological constraints by requiring ontological models to satisfy corresponding operational conditions on the ontological level, i.e. conditioned on complete access to (or fine-grained control over) the ontic state of the physical system. For instance, the operational non-signaling condition when conditioned on the ontic state of the shared physical system forms the primary ontological constraint of Bell's local-causality referred to as \textit{parameter independence}. However, the primary ontological constraints of these well-known ontological principles are (on their own) not enough to warrant consequences directly in contradiction with the predictions of quantum theory. Therefore, they invoke certain auxiliary constraints to enable a quantum violation. For instance, along with parameter independence, Bell's local-causality entails the auxiliary ontological assumption of outcome independence. The inclusion of these additional assumptions leads to dilution of the implication of corresponding quantum violations as they could be in principle attributed to the violation of the auxiliary assumptions leaving the primary assumptions intact. 

\textit{Measure zero operational condition:} One of the key issues with the experimental tests of these ontological principles
\cite{hensen2015loophole,kirchmair2009state,mazurek2016experimental},
lies in the associated operational conditions. Specifically, the well-known ontological notions of classicality, for their refutation, require associated zero-measure operational indistinguishability conditions to hold. Consequently, even the slightest of errors in the experimental composition may cause a deviation from the \text{zero measure} operational indistinguishability conditions. Besides, due to the finite size of experimental data the observed statistics almost always violate the zero measure operational conditions. This issue is often referred to as \textit{finite precision loophole}, and to address it for the aforementioned ontological principles, several schemes \cite{guhne2010compatibility,mazurek2016experimental,pusey2018robust,liang2019bounding,zhang2011asymptotically,lin2018device} have been proposed and implemented. 

In this work, we show that bounded ontological distinctness meets all of the aforementioned criteria for a fundamental notion of classicality. In particular, we demonstrate the violation of bounded ontological distinctness for preparations, and transformations, while employing three instances of two-dimensional quantum preparations and unitary transformations, respectively. Moreover, unlike the other well-known ontological principles, the primary ontological assumption of bounded ontological distinctness for preparations and transformations is enough to warrant a quantum violation on its own. This in turn leads to an unambiguous ontological implication, that quantum preparations and transformations are more ontologically distinct than they are operationally distinguishable. Moreover, we demonstrate that bounded ontological distinctness implies all other well-known notions of classicality. This provides a crucial unifying insight, namely, the quantum violations of the other ontological principles are, in essence, a demonstration of the implicate quantum excess ontological distinctness. Finally, the operational distinguishability condition accompanying bounded ontological distinctness is not a zero-measure condition, and serves an alternative approach for the resolution of the finite precision loophole. Note that, a similar approach for resolving the finite precision loophole for generalized noncontextuality was suggested by Spekkens in \cite{spekkens2005contextuality}, wherein, instead of the zero-measure operational equivalence condition, one would require a measure based operational similarity condition.

As Spekkens' generalized noncontextuality embodies the Leibniz principle \cite{spekkens2005contextuality,spekkens2019ontological}, and bounded ontological distinctness bases itself on a generalization of the Leibniz principle, bounded ontological distinctness can be viewed as an alternative (or a generalization) to Spekkens' noncontextuality. However, in contrast to generalized noncontextuality, the mathematical formulation bounded ontological distinctness requires a measure of operational distinguishability along with a corresponding measure of ontological distinctness. In this work, we have employed the maximum probability of (minimum error) discrimination as the measure of operational distinguishability and distinctness. As a consequence of this particular choice of measure, bounded ontological distinctness for preparations generalizes the notion of maximal $\psi$-epistemicity introduced in \cite{PhysRevLett.112.250403} from a pair of pure quantum preparations to an arbitrary number of pure or mixed quantum preparations. Apart from these, there are several other (relatively less related) results that deal with ontological explanation of the precise amount of operational distinguishability of preparations including \cite{harrigan2010einstein,leifer2013maximally,leifer2014quantum,PhysRevX.8.011015}. In particular, \cite{PhysRevX.8.011015} elicits the role of contextuality in quantum state discrimination.

\textit{Summary:} We begin by detailing the general formulation of operational theories and ontological models. In the section that follows, we formalize the definitions of distinguishability of sets of operational physical entities and distinctness of sets of corresponding ontological entities. While as a measure of distinguishability we use the maximal probability of distinguishing the constituent physical entities out of a uniform ensemble, as a measure of distinctness we employ the maximal probability of distinguishing the corresponding ontological entities when sampled from a uniform ensemble (conditioned on fine-grained control and access to the ontic state). 

Next, we formulate the principle of bounded ontological distinctness for sets of physical entities, which equates their distinguishability to the distinctness of corresponding ontological entities (Definitions \ref{BOD}). This principle leads to certain operational consequences. For the specific case of two pure quantum preparations, we note that bounded ontological distinctness for preparations coincides with the definition of maximal $\psi$-epistemicty based on the symmetric overlap of epistemic states \cite{PhysRevLett.112.250403}. Later, we observe that for any operational fragment of quantum theory entailing two pure preparations, there exits a maximally $\psi$-epistemic model that satisfies bounded ontological distinctness for preparations (Observation \ref{KSnogoPair}). 

Moving on, based on the distinguishability of three preparations we obtain an inequality, specifically, an upper-bound on the average pair-wise distinguishability of these preparations, valid in all operational theories that admit ontological models that adhere to bounded ontological distinctness of preparations (Proposition \ref{propositionBOD}). Subsequently, we present a set of three two dimensional quantum states which violate this inequality thereby demonstrating excess ontological distinctness of quantum preparations (FIG. \ref{PrepMeas1}). As maximal distinguishability forms an intrinsic property of the set of quantum preparations, this result presents a novel perspective, the conflict with bounded ontological distinctness lies in the relation between two intrinsic properties of the set of quantum preparations under consideration, namely, the maximum distinguishability and the maximum average pair-wise distinguishability. Moreover, this formulation allows us to infer a lower bound on the extent of excess ontological distinctness of the epistemic states underlying the set of quantum preparations. For a specific set of three two-dimensional quantum states, we observe that the Kochen-Specker ontological model saturates this lower bound (Observation \ref{KSpropositionBOD}). 

Additionally, for four preparations, the distinguishability of a pair of disjoint two-preparation mixtures yields an upper-bound on the average distinguishability of the remaining disjoint pairs of two-preparation mixtures, valid in all operational theories that admit ontological models adhering to bounded ontological distinctness of preparations along with the preservation of convexity of operational preparations on the ontic level (Proposition \ref{propostionRhoEpistemic}). This forms a robust version of the preparation noncontextual inequality featured in \cite{spekkens2009preparation}, and is a direct implementation of the aforementioned suggestion of Spekkens in \cite{spekkens2005contextuality} to employ an operational similarity condition instead of the zero-measure operational equivalence condition. Employing four two-dimensional quantum preparations, we demonstrate a quantum violation of this upper-bound (FIG. \ref{PrepMeas2}). 
Yet again, our methodology allows the inference of a tight lower bound on the extent of excess ontological distinctness of the mixed quantum preparations under consideration (Observation \ref{KSrhoepistemic}).  

Moving on, we demonstrate that unlike the case of preparations, it is not possible to violate bounded ontological distinctness of measurements in quantum theory. However, we show that if a pair of two-dimensional quantum measurements are neither completely indistinguishable nor perfectly distinguishable then outcome-deterministic ontological models violate bounded distinctness of measurements. 

Next, employing the distinguishability of a set of three transformations we obtain an operational inequality, valid in any theory which admits ontological models that adhere to bounded ontological distinctness of transformation (Proposition \ref{propositionBODT}). Subsequently, employing a set of three two-dimensional quantum unitary transformations we show that quantum theory violates this inequality, thereby demonstrating excess ontological distinctness of quantum transformations.

 In the following section, we present two distinct classes of communication tasks that benefit from excess ontological distinguishability of preparations. The first class of communication tasks constrains the distinguishability of the sender's preparations. As an example, we present a communication task and obtain a bound on the maximum success probability of this task achievable when employing preparations pertaining to operational theories that satisfy bounded ontological distinctness (Proposition \ref{propositionBOD2}). Subsequently, we provide an advantageous quantum protocol entailing three two-dimensional preparations and two binary outcome measurements, which violates this bound (FIG. \ref{PrepMeas4}). In the second class of communication tasks, we tolerate a bounded amount of leakage of information about the sender's private information. As an example, we consider the parity oblivious multiplexing \cite{spekkens2009preparation} with a bounded amount of leakage of information regarding the parity bit. Moreover, we demonstrate that in general classical models, and in particular classical $d$-level communication adheres to bounded ontological distinctness of preparations substantiating its candidature as a notion of classicality. 

In the penultimate section, we bring forth the fact that bounded ontological distinctness implies other ontological notions of classicality such as Spekkens' noncontextuality, Kochen-Specker noncontextuality and Bell's local-causality, which in turn substantiates the insight that excess ontological distinctness of quantum theory is implicate in the violations of these ontological principles. In particular, we show that generalized noncontextuality is a special case of bounded ontological distinctness. Finally, we proof the strict unidirectionality of this implication with the aid of a class of (universally) noncontextual fragment of quantum theory that allows for operational violations of the implications of bounded ontological distinctness (Observation \ref{impObservation}).

Finally, we conclude by laying out certain implications of our results, key conceptual insights, and tentative avenues for future investigation.

\section{Preliminaries}
In this section, we lay down the specifics of the experimental scenario employed in this work and revisit the general framework for operational theories and underlying ontological models.

\subsection{Prepare, transform and measure experiments} 
In prepare, transform and measure experiments, preparation of a physical system is followed by a transformation and finally, a measurement on the same. In each run of the experiment, a preparation, a transformation and a measurement are chosen and the outcomes of the measurement are recorded. The whole process is repeated several times to obtain frequency statistics. Crucially, each run of the experiment is assumed to be statistically independent. However, for the majority of this work, we shall be concerned with prepare and measure experiments, wherein the preparation of a physical system is immediately followed by a measurement on the same.\\

\subsection{Operational theories}
An operational interpretation of a physical theory serves a two-fold purpose: $(i)$ \textit{prescription:} specification of preparation procedures $P$, transformation procedures $T$ and measurement procedures $M$ and, $(ii)$ \textit{prediction} of probabilities $p(k|P,T,M)$ of obtaining an outcome $k$ given a measurement $M$ was performed on a preparation $P$ which underwent a transformation $T$. \\
For instance, quantum formalism prescribes a density matrix $\rho$ to model a preparation, a completely positive map to model a transformation and a positive-operator valued measure (POVM) $\{M_k\}$ to model a measurement. However, in this work we need only invoke unitary quantum transformations $U$. 
The probability of obtaining an outcome $k$ is given by the so called ``Born rule": $p(k|P,T,M)=\Tr{U\rho U^\dagger M_k}$. The operational predictions in prepare and measure experiments are of the form $p(k|P,M)$. Quantum predictions in such experiments are given by: $p(k|P,M)=\Tr{\rho M_k}$. 
\subsection{Ontological models} 
An ontological model seeks to explain the predictions of an operational theory, whilst assuming the existence of observer independent attributes associated with physical systems. These attributes describe the  ``real state of affairs" of a physical system \cite{spekkens2005contextuality,leifer2014quantum}. The complete specification of each such attribute is referred to as the \textit{ontic state} $\lambda$ of a physical system. The space of values that $\lambda$ may take is referred to as the \textit{ontic state space} $\Lambda$. Upon a preparation the physical system occupies a specific ontic state $\lambda\in\Lambda$. Crucially, in general, operational preparations may not allow for fine-grained operational control over the ontic state the system actually occupies i.e., an operational preparation procedure $P$, might only specify the probabilities $\mu(\lambda|P)$ of the system being in different ontic states. 
This in turn requires the ontic state space $\Lambda$ to be a measurable space, with a $\sigma$-algebra $\Sigma$ along with $\mu:\Sigma\rightarrow[0,1]$ being a $\sigma$-additive function which satisfies $\mu(\Lambda)=1$. Summarizing, in an ontological model, every preparation procedure $P$ is associated with an epistemic state (a normalized probability density over the ontic state space) $\mu(\lambda|P)$  such that $\int_{\Lambda}\mu(\lambda|P)d\lambda=1$. \\
A measurement in an ontological model simply measures the ontic state of a physical system. Consequently, ontological measurements are modelled by conditional probability distributions $\{\xi(k|\lambda)\}$ which specify the probability of obtaining the outcome $k$ given the ontic state $\lambda$. These probability distributions are measurable functions on $\Lambda$ and satisfy: \textit{positivity:} $\forall \ k,\lambda: \xi(k|\lambda) \geq 0$, \textit{completeness:} $\sum_{k}\xi(k|\lambda) =1$ and are referred to as \textit{response schemes}. However, the outcomes of an operational measurement procedure might not unveil the ontic state uniquely nor even allow for an inference of a set of ontic states within which the actual ontic state lies. In other words, the outcome of a measurement procedure might depend on the ontic state on only stochastically due to fundamental in-determinism of nature or dependence of measurement outcome on certain degrees of freedom other than the ontic state of the system being measured.
Consequently, every operational measurement procedure $M$ is associated with a response scheme $\{\xi(k|\lambda,M)\}$ which specifies the probability of obtaining the outcome $k$ given a measurement $M$ was performed on a physical system occupying the ontic state $\lambda$.
Averaging the conditional probabilities $\xi(k|\lambda,M)$ over our ignorance of the ontic state $\lambda$ yields the predictions of the ontological model for prepare and measure experiments, i.e., 
\be \label{pPM-ontic} 
p(k|P,M)=\int_{\Lambda}\mu(\lambda|P)\xi(k|\lambda,M)d\lambda . 
 \ee 
Crucially, the response schemes $\{\xi(k|\lambda,M)\}$ associated with operational measurements $M$ form a subset of the set of general response schemes $\{\xi(k|\lambda)\}$ that are constrained only by \textit{positivity} and \textit{completeness}.\\

A transformation in an ontological model simply alters the ontic state of the system. However, an operational transformation procedure may do so only stochastically. To accommodate the same, ontological transformations are modelled by transition schemes $\{\gamma(\lambda'|\lambda)\}$,  which specifies the probability of transition from the ontic-state $\lambda$ to the ontic-state $\lambda'$.  Summarizing, in a ontological model every operational transformation $T$ is associated with a transition scheme $\{\gamma(\lambda'|\lambda,T)\}$ such that $\forall \ \lambda: \int_\Lambda\gamma(\lambda'|\lambda,T)d\lambda'=1$. Finally, the predictions of the ontological model in prepare, transform and measure experiments are given by,
\be
p(k|P,T,M)=\int_{\Lambda^2}\mu(\lambda|P)\gamma(\lambda'|\lambda,T)\xi(k|\lambda',M) d\lambda d\lambda'.
\ee 

\section{Bounded ontological distinctness and quantum violation}
The ontological principle, ``bounded ontological distinctness" is based on the natural generalization of the Leibniz's principle of ``identity of indiscernibles" which states that maximal operational distinguishability reflects ontological distinctness. In this section, we invoke operational measures of distinguishability for sets of preparations, measurements and transformation and ontological measures of distinctness for sets of epistemic states, response schemes and transition schemes. Subsequently, we mathematically substantiate the principle of bounded ontological distinctness for the respective physical entities, derive certain consequences and report on quantum violations of these consequences. 
\subsection{Preparations}
\begin{figure}     
    \centering
    \includegraphics[scale=0.4]{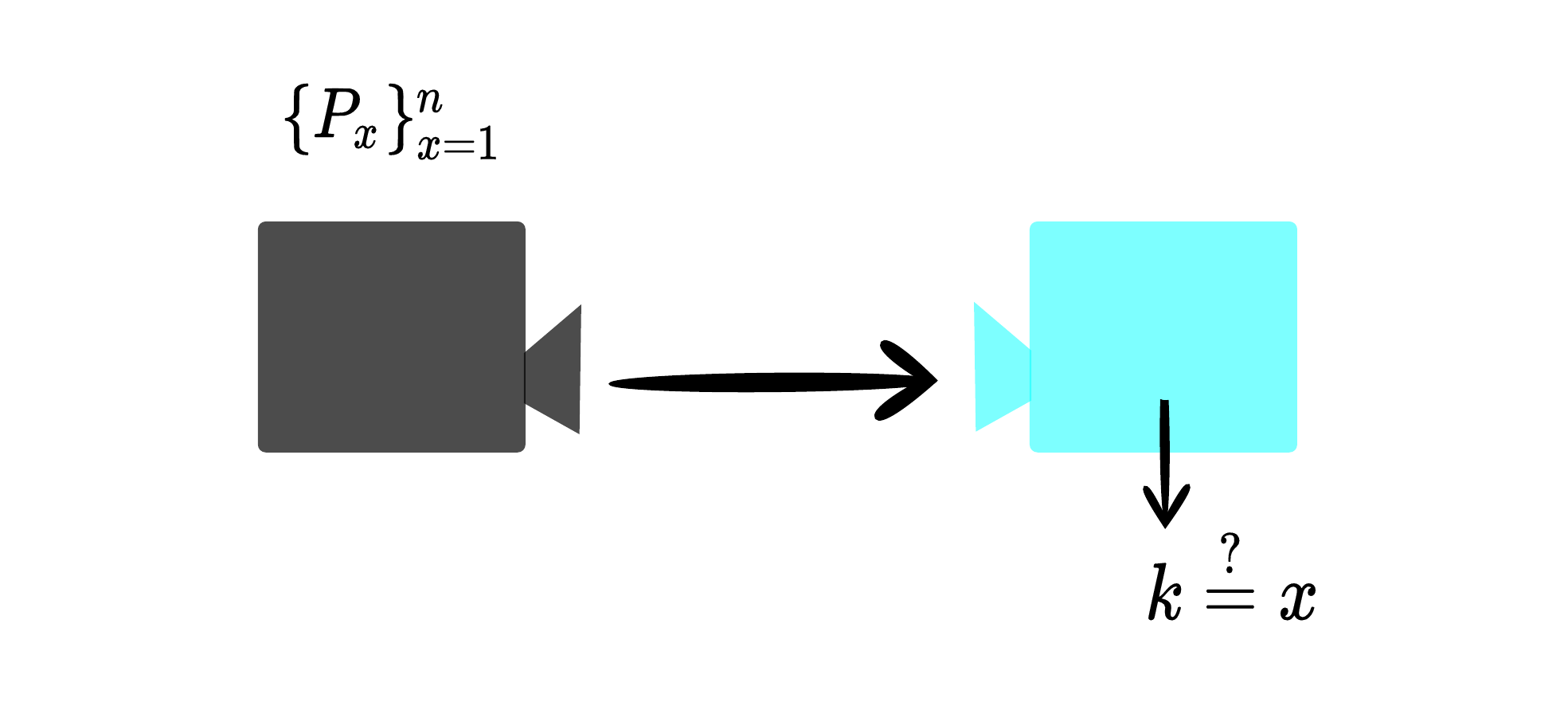}
    \caption{\label{bodp} This figure denotes the set-up for distinguishing a given set of preparations $\{P_x\}^n_{x=1}$ (dark black device) out of a uniform ensemble. In each run of the experiment, a constituent preparation $P_x$ is prepared, followed by a $n$-outcome measurement producing an outcome $k$. The (single-shot) distinguishing probability is simply the probability of producing an output $k=x$ as a result of a $n$-outcome measurement. This probability is maximized over all such measurements (light cyan device) yielding the distinguishability $s_\mathcal{O}$ of the constituent preparations. As a consequence of this maximization, $s_\mathcal{O}$ is an intrinsic property of this set of preparations and forms the operational condition accompanying bounded ontological distinctness of preparations. }
\end{figure}
First, we consider sets of preparations. We characterize these sets based on how well the constituent preparations can be operationally distinguished (FIG. \ref{bodp}). 
\begin{Definition} \textit{$p$-distinguishable preparations:}
In an operational theory, a set of preparations $\mathcal{P}\equiv \{P_x\}^n_{x=1}$ is termed \textit{$p$-distinguishable} if the constituent preparations can be distinguished when sampled from a uniform ensemble with at-most $p$ probability, i.e.,
\beq \label{pDistinguish}
s_{\mathcal{O}}= \max_{M}\Bigg\{\frac{1}{n} \sum_{x}p(k=x|P_x,M)\Bigg\}= p,
\eeq
 where $s_{\mathcal{O}}$ denotes the maximum operational probability of distinguishing these preparations when sampled from a uniform ensemble and the maximization is over the set of all possible $n$-outcome  measurements in a given operational theory, $M$ is an instance of such a measurement, and $k$ is an outcome of the measurement $M$.
\end{Definition}
 Observe that the maximum probability of distinguishing the preparations $s_{\mathcal{O}}$, is maximized over all possible measurements available in an operational theory. This maximization relieves $s_{\mathcal{O}}$ of its dependence on measurements, deeming it to be a suitable characterizing feature for sets of preparations. In quantum theory, the maximum probability of distinguishing for a set of quantum states $\mathcal{P}_Q$ out of a uniform ensemble has the expression,  
\begin{align} \label{pQuantumDistinguish}
    s_Q = \max_{\{M_k\}}\Bigg\{\frac{1}{n} \sum_{x}\Tr{\rho_xM_{k=x}}\Bigg\},
\end{align}
where the maximization is over the set of all possible $n$-outcome POVM $\{M_k\}$. For a specific set of quantum states $\mathcal{P}_Q$ this maximization can be cast as an efficient semi-definite program which essentially yields the maximum distinguishing probability $s_Q$, along with the optimal POVM $\{M_k\}$ \cite{wengang2008minimum,boyd2004convex}. \\ 
Each set of operational preparations $\mathcal{P}$ is associated with a set of epistemic states $\mathcal{P}_\Lambda$. We characterize these sets based on how distinct the constituent epistemic states actually are. 
\begin{Definition} \textit{$p$-distinct epistemic states:}
In an ontological model, a set of epistemic states $\mathcal{P}_\Lambda\equiv \{\mu(\lambda|P_x)\}^n_{x=1}$ is termed $p$-distinct if the constituent epistemic states can be ontologically distinguished (upon having access to the ontic state $\lambda$) when sampled from a uniform ensemble with at-most $p$ probability, i.e.,
\be \label{pLambdaDistinguish}
s_{\Lambda} = \max_{\{\xi(k|\lambda)\}}\Bigg\{\frac{1}{n} \sum_{x}\int_{\Lambda}\mu(\lambda|P_x)\xi(k=x|\lambda)d\lambda \Bigg\} =p ,
\ee
where $s_{\Lambda}$ denotes the maximum probability of distinguishing these epistemic states when sampled from a uniform ensemble given the ontic state $\lambda$, $\Lambda$ represents the ontic state space, $\mu(\lambda|P_x)$ is the epistemic state underlying the preparation $P_x$ and the maximization is over all valid response schemes $\{\xi(k|\lambda)\}$ which satisfy \textit{positivity} and \textit{completeness}.
\end{Definition}
Observe that the set of all possible $n$-outcome response schemes forms a convex polytope for each $\lambda$, the extremal points of this polytope are deterministic response schemes, i.e., each extremal response scheme is of the form: for each $\lambda$, $\xi(k|\lambda)=0$ except for specific $k=k_\lambda$ for which $\xi(k_\lambda|\lambda)=1$. In light of this observation, we can readily solve the above maximization by finding out the optimal extremal response scheme. Clearly, for each $\lambda$, the optimal response would be to output the index $x$ of the preparation which assigns the largest probability to $\lambda$. This in-turn leads us to the following succinct expression for $s_\Lambda$,
\beq \label{pLambdaDistinguishSuccinct}
s_{\Lambda} = \frac{1}{n}\int_\Lambda \max_{x}\bigg\{ \mu(\lambda|P_x)\bigg\}d\lambda.
\eeq
This expression further substantiates the fact that the maximization over response schemes relieves $s_\Lambda$ from its dependence on response schemes, deeming it a suitable characterizing feature of the set of epistemic states under consideration and a measure of the actual distinctness of its constituents. \\
If an ontological model explains the predictions of an operational theory then for a set of preparations $\mathcal{P}$ one may readily re-express  the maximum probability of operationally distinguishing these preparations using \eqref{pPM-ontic} when sampled from a uniform ensemble as, 
\be \nonumber
s_{\mathcal{O}} = \max_{\{\xi(k|\lambda,M)\}}\left\{\frac{1}{n} \sum_{x}\int_{\Lambda}\mu(\lambda|P_x)\xi(k=x|\lambda,M)d\lambda \right\},
\ee
where the maximization is over only the response schemes $\{\xi(k|\lambda,M)\}$ associated with operational measurements $M$. 
Observe that as the operational measurements may not reveal the ontic state $\lambda$ and, having access to $\lambda$ may only enhance the ability to distinguish the preparations under consideration. Therefore, in general we have the inequality $s_\mathcal{O}\leqslant s_\Lambda$. The principle of bounded ontological distinctness for preparations requires the operational distinguishability of preparations to be completely explained by the ontological distinctness of the underlying epistemic states, and can be formulated as follows,

\begin{Definition}
Bounded ontological distinctness for preparations ($BOD_P$): The set of epistemic states $\mathcal{P}_\Lambda \equiv \{\mu(\lambda|P_x)\}^n_{x=1}$ underlying a set of $p$-distinguishable preparations $\mathcal{P}\equiv\{P_x\}^n_{x=1}$ is $p$-distinct, i.e., $s_{\Lambda} = s_{\mathcal{O}}=p$, yielding the ontological constraint,
\beq \label{BOD}
 \frac{1}{n}\int_\Lambda \max_{x }\bigg\{\mu(\lambda|P_x)\bigg\}d\lambda=p.
\eeq
\end{Definition}
When formulated in this way, $BOD_P$ serves as a criterion for characterization of ontological models of a given operational theory. It is useful, at this point, to define a criterion for characterization of operational theories, namely, an operational theory or a fragment thereof is said to satisfy $BOD_P$, if there exists an ontological model which satisfies $BOD_P$ for all sets of prescribed operational preparations. Conversely, an operational theory or a fragment thereof is said to violate $BOD_P$ if there exists no ontological model which satisfies $BOD_P$ for all sets of prescribed operational preparations. \\  
$BOD_P$ serves to connect an exclusively operational phenomenon (distinguishability of operational preparations) to an exclusively ontological property (distinctness of epistemic states), consequently leading to interesting implications. \\
For a pair of $p$-distinguishable preparations $\{P_1,P_2\}$, $BOD_P$ requires the underlying epistemic states $\{\mu(\lambda|P_1),\mu(\lambda|P_2)\}$ to be $p$-distinct, i.e.,
\beq \nonumber
& s_\Lambda & = \frac{1}{2}\int_\Lambda \max_{x \in \{1,2\} }\bigg\{\mu(\lambda|P_x)\bigg\}d\lambda \\ \nonumber
& &=1-\frac{1}{2}\int_{\Lambda} \min_{x \in \{1,2\} }\bigg\{\mu(\lambda|P_x)\bigg\}d\lambda \\ \nonumber
& &=s_{\mathcal{O}}=p,
\eeq
where the second equality follows from the fact that $\forall \ (a\geqslant0,b\geqslant0): \ \max\{a,b\} = a+b-\min\{a,b\}$. 
In particular, the term $(1-s_\Lambda)=\frac{1}{2}\int_{\Lambda} \min_{x \in \{1,2\} }\{\mu(\lambda|P_x)\}d\lambda$ is referred to as the (symmetric) overlap of the pair of epistemic states. Furthermore, for any pair of pure quantum preparations $\{P_1,P_2\}$ corresponding to the pure states $\{\ket{\psi_1},\ket{\psi_2}\}$, the maximum probability of distinguishing them from a uniform ensemble has the expression $s_{Q} = 1-\frac{1}{2}(1- \sqrt{1-|\langle \psi_1|\psi_2 \rangle|^2})$. Notice that, the ontological notion of (symmetric) maximal $\psi$-epistemicity \cite{PhysRevLett.112.250403} requires the (symmetric) overlap of the underlying epistemic states $(1-s_{\Lambda})$ to be large enough to explain the indistinguishability $(1-s_{Q})$ of any given pair of pure quantum preparations.
Consequently, for the special case of pairs of pure quantum preparations $BOD_P$ is equivalent to the notion of (symmetric) maximal $\psi$-epistemicity, i.e. any ontological model which adheres to $BOD_P$ for a given pair of pure quantum preparations must be (symmetric) maximally $\psi$-epistemic and vice-versa.  However, as we shall demonstrate in Observation \ref{KSnogoPair}, there exists a ontological model which satisfies $BOD_P$ (or equivalently satisfies maximal $\psi$-epistemicity) for any prepare and measure fragment of quantum theory entailing just two pure quantum preparations along with their convex mixtures. \\

\subsubsection{Bounded ontological distinctness for three preparations.}
Moving on, we explore another implication of $BOD_P$ via the following proposition, wherein we employ $BOD_P$ for a set of three preparations to arrive at an upper bound on operational average pairwise distinguishability of these preparations, which is then violated by a prepare and measure fragment of quantum theory entailing a set of three pure two-dimensional preparations.
\begin{Proposition} \label{propositionBOD} If an operational theory admits ontological models adhering to $BOD_P$, then a set of three $p$-distinguishable preparations $\mathcal{P}\equiv\{P_1,P_2,P_3\}$, are pairwise distinguishable with the average probability of success being at-most $\frac{1+p}{2}$ , i.e.,
\beq \label{BODineq}
\frac{1}{3}\bigg(s^{1,2}_{\mathcal{O}}+s^{2,3}_{\mathcal{O}}+s^{3,1}_{\mathcal{O}}\bigg)\leqslant \frac{1+p}{2},
\eeq
where $s^{i,j}_{\mathcal{O}}$ denotes the maximum probability of distinguishing the preparations $\{P_i,P_j\}$ out of a uniform ensemble, i.e., $s^{i,j}_\mathcal{O}=\max_{M}\{\frac{1}{2} \sum_{x\in \{i,j\}}p(k=x|P_x,M)\}$.   
\end{Proposition}
\begin{proof} 
We proceed in two steps, $(i)$ we obtain an upper bound on the maximal average pairwise distinctness for a set three epistemic states in terms of maximal ontological distinctness and, $(ii)$ we employ $BOD_P$ to port this ontological relation to distinguishability of operational preparations.
Recall that, for a set of three epistemic states $\mathcal{P}_\Lambda\equiv\{\mu(\lambda|P_1),\mu(\lambda|P_2),\mu(\lambda|P_3)\}$, the maximal probability of distinguishing them when sampled from a uniform ensemble, upon having access to the ontic state $\lambda$ has the expression,
\beq \label{pLambdaDistinguish3}\nonumber
&s_{\Lambda} & =  \frac{1}{3}\int_\Lambda \max_{x }\bigg\{\mu(\lambda|P_x)\bigg\}d\lambda \\  \nonumber
& & = 1-\frac{1}{3}\int_\Lambda \min\bigg\{\mu(\lambda|P_1)+\mu(\lambda|P_2),\mu(\lambda|P_2)+\mu(\lambda|P_3), \\
& & \hspace{75pt} \mu(\lambda|P_3)+\mu(\lambda|P_1)\bigg\}
d\lambda,
\eeq
where for the second equality we employed the fact that $\forall \ (a\geqslant0,b\geqslant0,c\geqslant0): \ \max\{a,b,c\} = a+b+c-\min\{a+b,b+c,c+a\}$ along with  the property $\forall x\in\{1,2,3\}: \  \int_\Lambda\mu(\lambda|P_x)d\lambda = 1$. Similarly, for the pair of epistemic states $\{\mu(\lambda|P_1),\mu(\lambda|P_2)\}$, the maximal probability of distinguishing them when sampled from a uniform ensemble, upon having access to the ontic state $\lambda$ can be expressed as follows,
\beq \label{pLambdaDistinguish2} \nonumber
&s^{1,2}_{\Lambda} & =  \frac{1}{2}\int_\Lambda \max\bigg\{\mu(\lambda|P_1),\{\mu(\lambda|P_2)\bigg\}d\lambda \\ 
& & = 1-\frac{1}{2}\int_\Lambda \min\bigg\{\mu(\lambda|P_1),\mu(\lambda|P_2)\bigg\}
d\lambda.
\eeq
Similar expressions can be obtained for $s^{2,3}_\Lambda$ and $s^{3,1}_\Lambda$, so as to arrive at the expression for maximal average pairwise distinguishing probability for the three epistemic states under consideration,
\begin{widetext}
\begin{IEEEeqnarray}{rCl} \label{preBODineq1} \nonumber
& \frac{1}{3}\bigg(s^{1,2}_{\Lambda}+s^{2,3}_{\Lambda}+s^{3,1}_{\Lambda}\bigg)& = 1-\frac{1}{6}\int_\Lambda  \Bigg( \underbrace{ \min\bigg\{\mu(\lambda|P_1),\mu(\lambda|P_2)\bigg\}+\min\bigg\{\mu(\lambda|P_2),\mu(\lambda|P_3)\bigg\}+\min\bigg\{\mu(\lambda|P_3),\mu(\lambda|P_1)\bigg\} }_{ {\textstyle
  \geqslant \min\bigg\{\mu(\lambda|P_1)+\mu(\lambda|P_2),\mu(\lambda|P_2)+\mu(\lambda|P_3),\mu(\lambda|P_3)+\mu(\lambda|P_1)\bigg\} } } \Bigg) d\lambda \\ 
& & \leqslant \frac{1+s_\Lambda}{2},
\end{IEEEeqnarray}
\end{widetext}

where the inequality follows from the fact that $\forall \ (a\geqslant0,b\geqslant0,c\geqslant0): \ \min\{a,b\}+\min\{b,c\}+\min\{c,a\}=\min\{a+b,b+c,c+a\}+\min\{a,b,c\}  $ and \eqref{pLambdaDistinguish3}. Consequently, the inequality \eqref{preBODineq1} is saturated when $\forall \lambda : \ \min\{\mu(\lambda|P_1),\mu(\lambda|P_2),\mu(\lambda|P_3)\}=0$. As in general $s^{i,j}_{\mathcal{O}}\leqslant s^{i,j}_{\Lambda}$, we have,
\begin{align} \label{preBODineq2}
\nonumber
\frac{1}{3}\bigg(s^{1,2}_{\mathcal{O}}+s^{2,3}_{\mathcal{O}}+s^{3,1}_{\mathcal{O}}\bigg)
& \leqslant \frac{1}{3}\bigg(s^{1,2}_{\Lambda}+s^{2,3}_{\Lambda}+s^{3,1}_{\Lambda}\bigg) \\
& \leqslant \frac{1+s_\Lambda}{2}
\end{align}
Now, if an operational theory admits ontological models that adhere to bounded onotological distinguishability, every set of $p$-distinguishable preparations is associated with a set of $p$-distinct epistemic states i.e. $s_\Lambda=s_{\mathcal{O}}=p$. Inserting this into \eqref{preBODineq2} yields the desired thesis.
\end{proof}

\subsubsection{Excess ontological distinctness of three quantum preparations} \label{markSection1}
 We are now prepared to demonstrate the quantum violation of $BOD_P$. Consider, three qubit preparations
 \be \label{rho123}
 \rho_i=\frac{\mathbb{I}+\vec{n}_i\cdot \vec{\sigma}}{2},  \quad i=1,2,3, \ee
 where $\vec{\sigma}=[\sigma_x,\sigma_y,\sigma_z]$ is a vector of the Pauli matrices, $\vec{n}_1=[1,0,0]^T,\vec{n}_2=[\cos{\frac{2\pi}{3}},\sin{\frac{2\pi}{3}},0]^T$ and $\vec{n}_3=[\cos{\frac{4\pi}{3}},\sin{\frac{4\pi}{3}},0]^T$ are Bloch vectors. 
Observe that for any three two dimensional quantum preparation the maximum probability of distinguishing them when sampled from a uniform ensemble is upper bounded as $s_Q=\max_{\{M_k\}}\{\frac{1}{3}\sum_{x \in \{1,2,3\}}\Tr{\rho_xM_{k=x}}\} \leqslant \frac{1}{3} \sum_x  \Tr\{ M_x\} = \frac{2}{3}\approx 0.667$ (FIG. \ref{PrepMeas1}). A straightforward semi-definite program yields the optimal measurement $\{M_x=\frac{2\rho_x}{3}\}$ which saturates this upper bound.  This implies that these states form a set of $\frac{2}{3}\approx 0.667$-distinguishable quantum preparations. It follows from Proposition \ref{propositionBOD} that if an operational theory satisfies $BOD_P$, then these preparations are pairwise distinguishable with at-most average probability $\frac{5}{6}\approx 0.833$, i.e., $\frac{1}{3}(s^{1,2}_{\mathcal{O}}+s^{2,3}_{\mathcal{O}}+s^{3,1}_{\mathcal{O}})\leqslant \frac{5}{6}\approx 0.833 $. \\
However, the maximum probability of distinguishing the pair of quantum states $\rho_1,\rho_2$ out of a uniform ensemble has the following succinct expression,
\begin{align} \label{pQuantumDistinguish2}
    s^{1,2}_Q =
    \frac{1}{2}\Bigg(1+T(\rho_1,\rho_2)\Bigg)=\frac{1}{2}\Bigg(1+\frac{\sqrt{3}}{2}\Bigg)\approx 0.933,
\end{align}
where $T(\rho_1,\rho_2)$ is the trace distance between the density matrices $\{\rho_1,\rho_2\}$, and $s^{1,2}_Q$ is the maximum probability for distinguishing quantum states $\{\rho_1,\rho_2\}$ out of a uniform ensemble. As the states under consideration are symmetrically distributed, the maximum probability of distinguishing $\{\rho_2,\rho_3\}$ and $\{\rho_3,\rho_1\}$ remains unchanged i.e. $s^{2,3}_Q=s^{3,1}_Q=s^{1,2}_Q=\frac{1}{2}(1+ \frac{\sqrt{3}}{2})\approx 0.933$. This in turn leads us to,
\begin{align} \label{BODviolation}
 \frac{1}{3}\bigg(s^{1,2}_{Q}+s^{2,3}_{Q}+s^{3,1}_{Q}\bigg) = \frac{1}{2}\Bigg(1+ \frac{\sqrt{3}}{2}\Bigg)\approx 0.933,  
\end{align}
which is a clear violation of Proposition \ref{propositionBOD} and consequently of $BOD_P$. This implies that any ontological model that seeks to explain the predictions of quantum theory must feature \textit{excess ontological distinctness}. Specifically, in order to explain the maximal average probability for pairwise distinguishing the aforementioned states the set of epistemic states underlying this set of $\frac{2}{3}\approx 0.667$-distinguishable qubits must be at-least $\frac{\sqrt{3}}{2} \approx 0.866$-distinct.
\begin{figure}     
    \centering
    \includegraphics[width=0.5\textwidth]{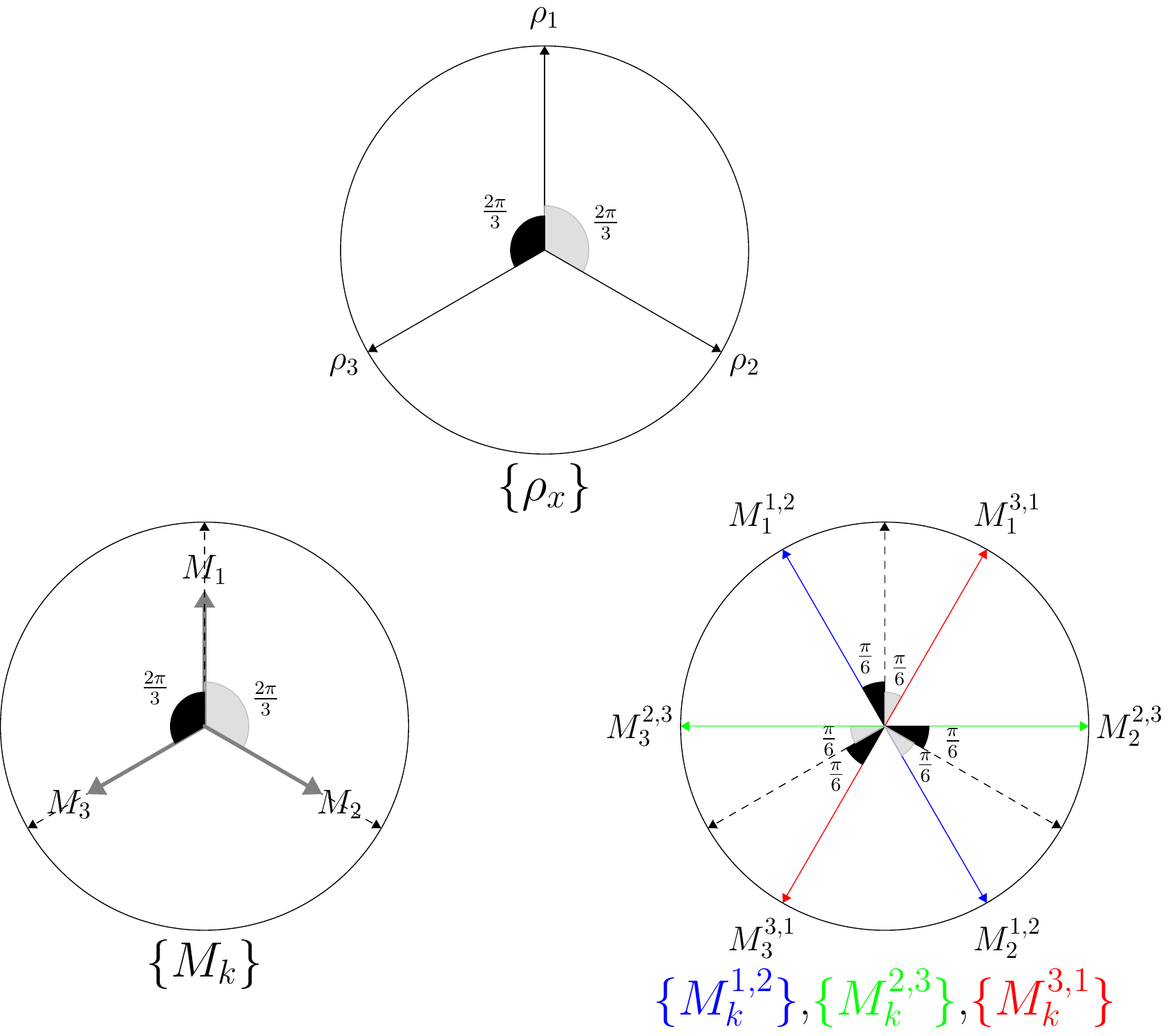}
    \caption{\label{PrepMeas1} This figure depicts the Bloch vectors in the x-y plane (x axis being vertical and y axis horizontal) of the Bloch sphere corresponding to the set of states $\{\rho_x=\frac{\mathbb{I}+\vec{n}_x\cdot \vec{\sigma}}{2}\}_{x\in\{1,2,3\}}$ (top) where $\vec{n}_1=[1,0,0]^T,\vec{n}_2=[\cos{\frac{2\pi}{3}},\sin{\frac{2\pi}{3}},0]^T$ and $\vec{n}_3=[\cos{\frac{4\pi}{3}},\sin{\frac{4\pi}{3}},0]^T$, the POVM $\{M_k\}_{k\in \{1,2,3\}}$ (bottom left) which optimally distinguishes these states such that $s_{Q}=\frac{2}{3}\approx 0.667$, and the projective measurements $\{M^{1,2}_k\}_{k\in\{1,2\}},\{M^{2,3}_k\}_{k\in\{2,3\}},\{M^{3,1}_k\}_{k\in\{3,1\}}$ (bottom right) which maximize the respective pair-wise distinguishability such that $s^{1,2}_Q=s^{1,2}_Q=s^{1,2}_Q=\frac{1}{2}(1+\frac{\sqrt{3}}{2})\approx 0.933$, leading to a violation of Proposition \ref{propositionBOD}. Note that the POVM elements $M_x=  \frac{\mathbb{I}+\vec{n}_x\cdot \vec{\sigma}}{3}$ (bottom left) do not correspond to a vector in the Bloch sphere, instead we have employed a common heuristic depiction scheme entailing short thick gray arrows to denote the POVM elements, emphasize their directions and normalization \cite{kurzynski2006graphical}.}
\end{figure}
Furthermore, we demonstrate extensive violation of $BOD_P$ by randomly sampling triplets of pure qubits and pure qutrits. Remarkably, $72\%$ of randomly picked triplets of pure qubits and $58\%$ of randomly picked triplets of pure qutrits (based on a uniform distribution on the unit hypersphere) exhibit excess ontological distinctness  (FIG. \ref{qubitQutritRandom}). Furthermore, over $46\%$ of randomly picked triplets of mixed qubit states and over $29\%$ of randomly picked triplets of mixed qutrit states (based on a uniform distribution on the unit hypersphere) exhibit excess ontological distinctness.  \\

\subsubsection{Examples of quantum ontological models}
While we have established the fact that any ontological model of quantum theory must exhibit excess ontological distinctness, now we discuss three specific examples \cite{harrigan2010einstein}: $(i)$ the Beltrametti-Bugajski ontological model posits that the different possible ontic-states are simply the different possible quantum states \cite{beltrametti1995classical}. As this model takes the quantum state
alone to be a complete description of reality, it is $\psi$-complete. As in $\psi$-complete models each pure quantum preparation is associated with a unique ontic state, all sets of pure non-identical quantum preparations have $s_{\Lambda}=1$. $(ii)$ Similarly, the Bell-Mermin ontological model for a two dimensional Hilbert space is a $\psi$-ontic model wherein every ontic state is associated with only one pure quantum preparation \cite{bell1966problem,mermin1993hidden}. Consequently, in $\psi$-onitc models all sets of pure non-identical quantum preparations have $s_{\Lambda}=1$. $(iii)$ The more intriguing $\psi$-epistemic ontological models posit ontic states that are consistent with more 
than one pure quantum state. A well known $\psi$-epistemic ontological model for a two dimensional Hilbert space was proposed by Kochen and Specker \cite{kochen1967problem}. In this model the ontic-state space $\Lambda$ is taken to be the unit sphere and each pure qubit preparation $P\equiv \rho=\frac{\mathbb{I}+\vec{n}\cdot\vec{\sigma}}{2}$ is associated with the epistemic state $\mu(\lambda|P)=\frac{1}{\pi}(H(\vec{n}\cdot\vec{\lambda})\vec{n}\cdot\vec{\lambda})$ where $H(x)=1$ if $x>0$,  and $H(x)=0$ if $x\leqslant 0$. Now, we present two intriguing observations pertaining to Kochen-Specker ontological model.
\begin{Observation} \label{KSpropositionBOD}
Remarkably, the epistemic states prescribed by the Kochen-Specker $\psi$-epistemic ontological model associated with the aforementioned quantum preparations $\{\rho_1,\rho_2,\rho_3\}$ \eqref{rho123} saturate the lower bound on the extent of excess ontological distinct inferred from the violation of \eqref{BODineq} i.e. this set is precisely $\frac{\sqrt{3}}{2}\approx 0.866$-distinct.
\end{Observation}
\begin{proof}
The maximum probability of distinguishing the epistemic-states associated with $\rho_1,\rho_2$ and $\rho_3$ upon having access to the ontic state $\lambda$ is,
\beq \label{KSThreeQubits} \nonumber
 s_\Lambda & =& \frac{1}{3\pi}\int_{\Lambda}\max_{x\in\{1,2,3\}}\Bigg\{H(\vec{n_x}\cdot\vec{\lambda})\vec{n_x}\cdot\vec{\lambda}\Bigg\}d\lambda\\ \nonumber
& =& \frac{1}{3\pi}\int_{\Lambda}\max\Bigg\{\lambda_x,\cos{\frac{2\pi}{3}}\lambda_x+\sin{\frac{2\pi}{3}}\lambda_y,  \\ 
&& \quad \quad \quad \cos{\frac{4\pi}{3}}\lambda_x+\sin{\frac{4\pi}{3}}\lambda_y\Bigg\}d\lambda \nonumber \\
& =& \frac{\sqrt{3}}{2},
\eeq
where for the second equality we use the fact that $(i)$ the ontic-state space is the surface of a unit sphere and the ontic state is a real three dimensional unit vector $\lambda=[\lambda_x,\lambda_y,\lambda_z]^T$ and $(ii)$ because of the particular orientation of the Bloch vectors of the quantum states under consideration, the function $H(\vec{n_x}\cdot\vec{\lambda})$ in the integrand is redundant as $\max\{\vec{n_1}\cdot\vec{\lambda},\vec{n_2}\cdot\vec{\lambda},\vec{n_3}\cdot\vec{\lambda}\}$ is never negative. 
\end{proof}
We have demonstrated violation of $BOD_P$ whilst employing triplets of qubits, however, when 
we restrict ourselves to just pairs of pure quantum preparations, we make the following intriguing observation,
\begin{Observation} \label{KSnogoPair}
Any prepare and measure fragment of quantum theory entailing a pair of pure preparations and all possible measurements satisfies $BOD_P$, i.e. there exists an ontological model for such fragments which adheres to $BOD_P$.
\end{Observation}
\begin{proof}
Observe that any two pure quantum states, without loss of generality, can be effectively represented as a pair of two-dimensional quantum states (which preserve the overlap). Now, for pairs of qubits Kochen-Specker ontological model satisfies $BOD_P$, i.e. for all pairs of qubits the maximum operational distinguishability $s_Q$ is exactly the ontological distinctness $s_\Lambda$ of the associated epistemic states in this model.
To see this, consider two pure qubits $\rho_1=\frac{\mathbb{I}+\vec{n_1}.\vec{\sigma}}{2},\rho_2=\frac{\mathbb{I}+\vec{n_2}.\vec{\sigma}}{2}$. Without loss of generality, we can assume the Bloch vector $\vec{n_1}$ associated with the first qubit to be $[1,0,0]^T$ and the Bloch vector $\vec{n}_2$ to be $[\cos{\theta_0},\sin{\theta_0},0]^T$ where $\theta_0\in[0,\frac{\pi}{2}]$. The maximum operational distinguishing probability for these two qubits turns out to be $s_{Q}=\frac{1}{2}(1+\sin{\frac{\theta_0}{2}})$. Now in the Kochen-Specker model the ontological distinctness of corresponding epistemic states turns out to be,
\begin{widetext}
\beq \label{KSMaxPsiSym} \nonumber
&s_{\Lambda} &= \frac{1}{2\pi}\int_{\Lambda}\max\Bigg\{H(\vec{n_1}\cdot\vec{\lambda})\vec{n_1}\cdot\vec{\lambda},H(\vec{n_2}\cdot\vec{\lambda})\vec{n_2}\cdot\vec{\lambda}\Bigg\}d\lambda \nonumber \\
& & = \frac{1}{2\pi}\int_{\Lambda}\max\Bigg\{0,\lambda_x,\cos{\theta_0}\lambda_x+\sin{\theta_0}\lambda_y\Bigg\}d\lambda \nonumber \\
& & = \frac{1}{2\pi}\int_{\theta_{\lambda}=-\pi}^{\pi}\int_{\phi_\lambda=0}^{\pi}\max\Bigg\{0,\cos{\theta_\lambda}\sin{\phi_\lambda},\cos{\theta_0}\cos{\theta_\lambda}\sin{\phi_\lambda}+\sin{{\theta_0}}\sin{\theta_\lambda}\sin{\phi_\lambda}\Bigg\}\sin{\phi_\lambda}d\phi_{\lambda}d\theta_\lambda \nonumber \\
& & = \frac{1}{\pi}\int_{\theta_{\lambda}=-\frac{\pi}{2}}^{\frac{\theta_0}{2}}\int_{\phi_\lambda=0}^{\pi}\cos{\theta_\lambda}\sin^2{\phi_\lambda}d\phi_{\lambda}d\theta_\lambda \nonumber \\
& & = \frac{1}{2}\left(1+\sin{\frac{\theta_0}{2}}\right)=s_Q,
\eeq
\end{widetext}
where for the second equality we used the fact that $(i)$ the ontic-state space is the surface of a unit sphere and the ontic state is a real three dimensional unit vector $\lambda=[\lambda_x,\lambda_y,\lambda_z]^T$ and $(ii)$ the effect of the function $H(.)$ in the original integral is to assign zero whenever the other two terms are negative. For the third equality we expressed the ontic state in standard spherical coordinates $[\lambda_x,\lambda_y,\lambda_z]^T\equiv[\cos{\theta_\lambda}\sin{\phi_\lambda},\sin{\theta_\lambda}\sin{\phi_\lambda},\cos{\phi_\lambda}]^T$ where $\theta\in[-
\pi,\pi]$ is the azimuthal and $\phi\in[0,\pi]$ is the polar spherical coordinate. The fourth equality follows from the observations that $(i)$ the second function in the maximization is symmetric with respect to $[1,0,0]^T$ on the azimuthal plane, $(ii)$ the third function in the maximization is simply the second function rotated by an angle $\theta_0$ on the azimuthal plane and, $(iii)$ the functions are mirror symmetric on either side of the plane bisecting $\theta_0$ and cross over on the same plane i.e.  $\forall \ \theta_\lambda\in[-\frac{\pi}{2},\frac{\theta}{2}],{\theta_\lambda}'=\theta_0-\theta_\lambda$, $\max\{0,\cos{\theta_\lambda},\cos{\theta_0}\cos{\theta_\lambda}+\sin{{\theta_0}}\sin{\theta_\lambda}\}=\max\{0,\cos{\theta_\lambda}',\cos{\theta_0}\cos{\theta_\lambda}'+\sin{{\theta_0}}\sin{\theta_\lambda}'\}=\cos{\theta_\lambda}$.  
\end{proof}

\subsubsection{Bounded ontological distinctness of two mixed preparations}
 With the aid of the following additional property of ontological models we demonstrate excess ontological distinctness by bounding the ontological distinctness of two mixed preparations. \\
 \begin{Definition} \textit{Convexity of epistemic states:} The epistemic state underlying an operational mixture of two preparations is the mixture of the respective underlying epistemic states (retaining the same proportions), i.e. if $P^c_{1+2}\equiv cP_1+(1-c)P_2$ then $\mu(\lambda|P^c_{1+2})=c\mu(\lambda|P_1)+(1-c)\mu(\lambda|P_2)$. 
 \end{Definition}
For the ease of notation, in what follows, we shall denote the uniform mixture as $P_{1+2}=P^{\frac{1}{2}}_{1+2}\equiv \frac{1}{2}(P_1+P_2)$ \footnote{We remark here that the ontological property: convexity of epistemic states, however natural, forms an additional assumption, for it is possible to have ontological models of quantum theory such as a $\rho-complete$ ontological model that do not exhibit this property.}.

In the following proposition, we consider four preparations and employ $BOD_P$ for a disjoint pair of two-preparation mixtures along with convexity of epistemic states to obtain an upper bound on average distinguishability of other disjoint pairs of two-preparation mixtures.
\begin{Proposition} \label{propostionRhoEpistemic}
Consider a set of four operational preparations $\{P_1,P_2,P_3,P_4\}$ where the set of mixtures $\{P_{1+2},P_{3+4}\}$ is $p$-distinguishable. Now, if an operational theory admits ontological models adhering to $BOD_P$ along with convexity of epistemic states, then the pairs of mixtures $\{P_{1+3},P_{2+4}\}$ and $\{P_{1+4},P_{2+3}\}$ are distinguishable with the average probability of success being at-most $\frac{1+p}{2}$ i.e. if $s^{1+2,3+4}_\mathcal{O}=p$ then,
\beq \label{rhoEpistemic}
\frac{1}{2}\bigg(s^{1+3,2+4}_{\mathcal{O}}+s^{1+4,2+3}_{\mathcal{O}}\bigg)\leqslant \frac{1+p}{2},
\eeq
where $s^{i+j,k+l}_{\mathcal{O}}$ denotes the maximum probability of distinguishing the preparations $\{P_{i+j},P_{k+l}\}$ out of a uniform ensemble.
\end{Proposition}
\begin{proof} 
We broadly follow the same steps as in the proof of Proposition \ref{propositionBOD}. 
Recall that, for a pair of epistemic states $\{\mu(\lambda|P_{i+j}),\mu(\lambda|P_{k+l})\}$, the maximal probability of distinguishing them when sampled from a uniform ensemble, upon having access to the ontic state $\lambda$ has the expression,
\beq \label{RhoEpistemic2} \nonumber
&s^{i+j,k+l}_{\Lambda} & =  \frac{1}{2}\int_\Lambda \max_{x\in \{{i+j},{k+l}\}}\bigg\{\mu(\lambda|P_x)\bigg\}d\lambda \\ \nonumber
& & = 1-\frac{1}{2}\int_\Lambda \min\bigg\{\mu(\lambda|P_{i+j}),\mu(\lambda|P_{k+l})\bigg\}
d\lambda \\ \nonumber
& & = 1-\frac{1}{4}\int_\Lambda \min\bigg\{\mu(\lambda|P_{i})+\mu(\lambda|P_{j}), \\ 
& & \hspace{75pt}
\mu(\lambda|P_{k})+\mu(\lambda|P_{l})\bigg\}d\lambda
\eeq 
where for the third equality we employed the assumption convexity of epistemic states. This in-turn leads us to the expression for average distinguishability for the pairs of epistemic states associated with pairs of preparations $\{P_{1+3},P_{2+4}\}$ and $\{P_{1+4},P_{2+3}\}$ upon having access to the ontic state $\lambda$,
\begin{widetext}
\beq \label{RhoEpistemic2Plus2} \nonumber
&\frac{1}{2}\bigg(s^{1+3,2+4}_{\Lambda}+s^{1+4,2+3}_{\Lambda}\bigg)=1-\frac{1}{8}\int_\Lambda\Bigg(&\min\bigg\{\mu(\lambda|P_1)+\mu(\lambda|P_3),\mu(\lambda|P_2)+\mu(\lambda|P_4)\bigg\}\\  
& &  + \min\bigg\{\mu(\lambda|P_1)+\mu(\lambda|P_4),\mu(\lambda|P_2)+\mu(\lambda|P_3) \bigg\}\Bigg)d\lambda. 
\eeq
The observation $\forall \ \{a\geqslant0,b\geqslant0,c\geqslant0,d\geqslant0\}, \ \min\{a+c,b+d\}+\min\{a+d,b+c\} \geqslant  \min\{a+b,c+d\}$ along with the fact that in general $s^{i+j,k+l}_{\mathcal{O}}\leqslant s^{i+j,k+l}_\Lambda$ yields,
\beq \label{PreRhoEpistemic}  \nonumber
&\frac{1}{2}\bigg(s^{1+3,2+4}_{\mathcal{O}}+s^{1+4,2+3}_{\mathcal{O}}\bigg)& \leqslant  1-\frac{1}{8}\int_\Lambda\min\bigg\{\mu(\lambda|P_1)+\mu(\lambda|P_2),\mu(\lambda|P_3)+\mu(\lambda|P_4)\bigg\}d \lambda \\ 
& & \leqslant \frac{1+s^{1+2,3+4}_{\Lambda}}{2} 
\eeq
\end{widetext}
Now, if an operational theory admits ontological models that adhere to bounded onotological distinguishability, every set $p$-distinguishable preparations is associated with a set of $p$-distinct epistemic states, in particular $s^{1+2,3+4}_{\Lambda}=s^{1+2,3+4}_{\mathcal{O}}=p$. Inserting this into \eqref{PreRhoEpistemic} yields the desired thesis.
\end{proof}
\subsubsection{Excess ontological distinctness of a pair of mixed quantum preparations}

\begin{figure}[]     
    \centering
    \includegraphics[width=0.45\textwidth]{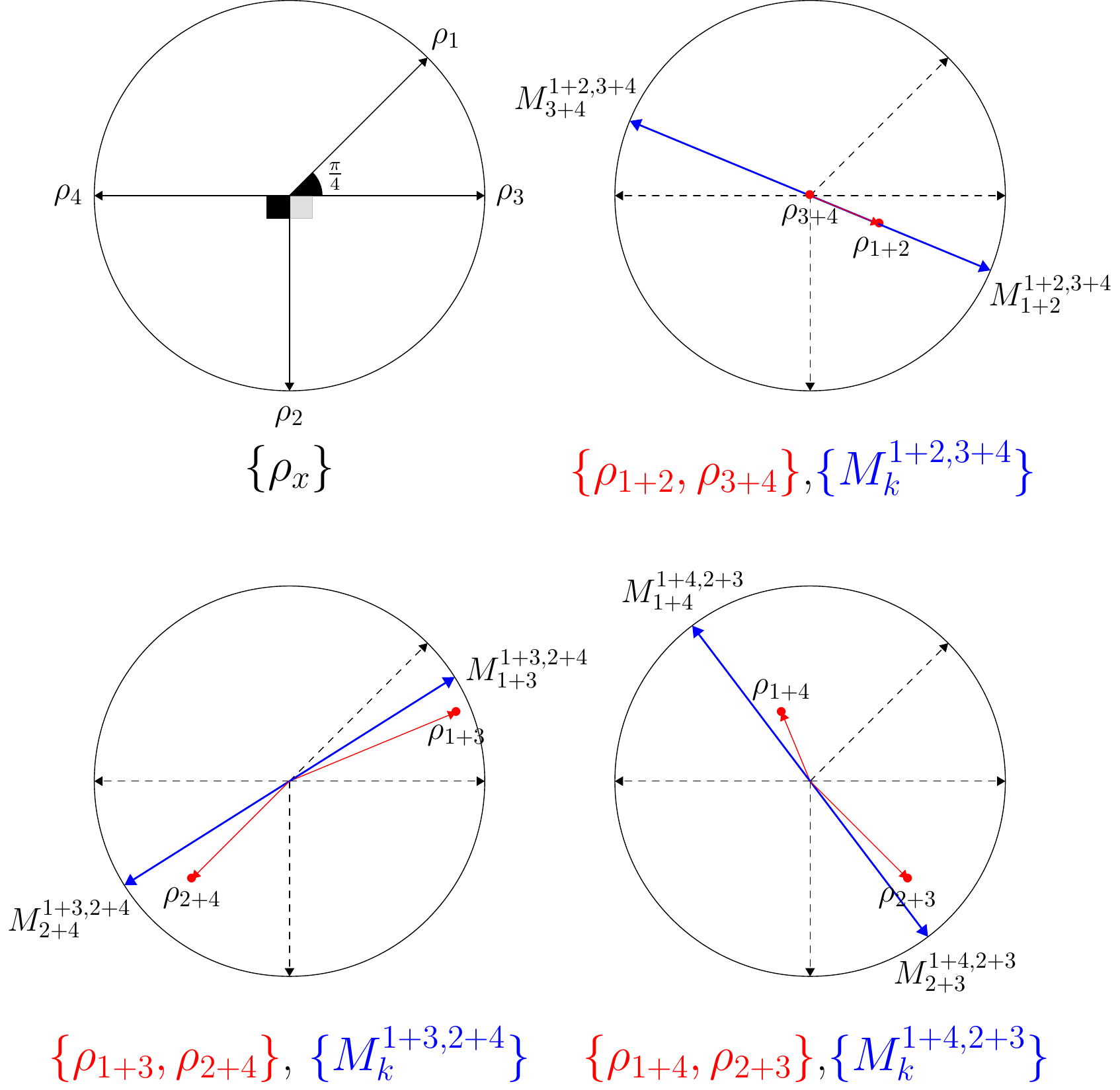}
    \caption{\label{PrepMeas2} This figure depicts the Bloch vectors in the x-y plane (x axis being vertical and y axis horizontal) of the Bloch sphere corresponding to the set of states $\{\rho_x=\frac{\mathbb{I}+\vec{n}_x.\vec{\sigma}}{2}\}_{x\in\{1,2,3,4\}}$ (top left) where $\vec{n}_1=[\cos{\theta},\sin{\theta},0]^T$, $\vec{n}_2=[-1,0,0]^T$, $\vec{n}_3=[0,1,0]^T$ and $\vec{n}_4=[0,-1,0]^T$ and $\theta = \frac{\pi}{4}$, the mixed preparations $\{\rho_{1+2},\rho_{3+4}\}$ along with the measurement $\{M^{1+2,3+4}_{k\in\{1+2,3+4\}}\}$ (top right) which maximizes their distinguishability such that $s^{1+2,3+4}_{Q}\approx 0.595$, the mixed preparations $\{\rho_{1+3},\rho_{2+4}\}$ along with the measurement $\{M^{1+3,2+4}_k\}_{k\in\{1+3,2+4\}}$ (bottom left) which maximizes their distinguishability such that $s^{1+3,2+4}_{Q}\approx 0.9$, and the mixed preparations $\{\rho_{1+4},\rho_{2+3}\}$ along with the measurement $\{M^{1+4,2+3}_{k\in\{1+4,2+3\}}\}$ (bottom right) which maximizes their distinguishability such that $s^{1+4,2+3}_{Q}\approx 0.767$, consequently leading to a violation of Proposition \ref{propostionRhoEpistemic}. As a consequence, the epistemic states underlying the $\approx0.595$-distinguishable quantum preparations $\{\rho_{1+2},\rho_{3+4}\}$ must be at-least $\approx0.667$-distinct.}
\end{figure}
To witness quantum violation of \eqref{rhoEpistemic}, consider four pure qubits of the form $\rho_x=\frac{\mathbb{I}+\vec{n}_x\cdot\vec{\sigma}}{2}$ where $\vec{n}_1=[\cos{\theta},\sin{\theta},0]^T$, $\vec{n}_2=[-1,0,0]^T$, $\vec{n}_3=[0,1,0]^T$ and $\vec{n}_4=[0,-1,0]^T$ and $\theta \in [0,\frac{\pi}{2}]$. The pair of mixtures $\{\rho_{1+2},\rho_{3+4}\}$ is $\frac{1}{2}(1+\frac{\sqrt{1-\cos{\theta}}}{2\sqrt{2}})$-distinguishable. It follows from Proposition \ref{propostionRhoEpistemic}, if an operational theory admits ontological models that adhere to $BOD_P$ and convexity of epistemic states then maximum average probability of distinguishing the pairs of mixtures $\{\rho_{1+3},\rho_{2+4}\}$ and $\{\rho_{1+4},\rho_{2+3}\}$ is upper-bounded by,
\be \frac{1}{2}(s^{1+3,2+4}_\mathcal{O} +s^{1+3,2+4}_\mathcal{O})\leqslant \frac{1}{16} \left(\sqrt{2} \sqrt{ 1-\cos (\theta ) }+12\right).\ee 
However the pairs of mixtures $\{\rho_{1+3},\rho_{2+4}\}$ and $\{\rho_{1+4},\rho_{2+3}\}$ are $\frac{1}{2}(1+\frac{\sqrt{3+\cos{\theta}+2\sin{\theta}}}{2\sqrt{2}})$-distinguishable and $\frac{1}{2}(1+\frac{\sqrt{3+\cos{\theta}-2\sin{\theta}}}{2\sqrt{2}})$-distinguishable respectively, which leads to violation of \eqref{rhoEpistemic} for $0\leqslant\theta\lessapprox\frac{\pi}{2\sqrt{2}}$ (FIG. \ref{fourQubit}), since 
\beq \label{refMe}
 \frac{1}{2}(s^{1+3,2+4}_Q+s^{1+4,2+3}_Q)&=& \frac{1}{2} + 
\frac{\sqrt{3+\cos{\theta}+2\sin{\theta}}}{8\sqrt{2}} \nonumber \\
&& +\frac{\sqrt{3+\cos{\theta}-2\sin{\theta}}}{8\sqrt{2}} .
\eeq 
In particular when $\theta_0=0$, to explain maximum average probability of distinguishing the pairs of mixtures $\{\rho_{1+3},\rho_{2+4}\}$ and $\{\rho_{1+4},\rho_{2+3}\}$,
$\frac{1}{2}(s^{1+3,2+4}_Q+s^{1+4,2+3}_Q)=\frac{2+\sqrt{2}}{4}\approx 0.853$
the epistemic states underlying completely indistinguishable mixtures $\{\rho_{1+2},\rho_{3+4}\}$ must be atleast $\frac{1}{\sqrt{2}}\approx 0.707$-distinct. 

\subsubsection{Examples of quantum ontological models}
In Beltrametti-Bugajski ontological model ($\psi$-complete) and the Bell-Mermin ontological model ($\psi$-ontic) and for the aforementioned states $\rho_1,\rho_2,\rho_3,\rho_4$ $s^{1+2,3+4}_\Lambda=1$. Yet again Kochen-Specker $\psi$-epistemic model stands out.
\begin{Observation} \label{KSrhoepistemic}
Remarkably, the epistemic states prescribed by the Kochen-Specker $\psi$-epistemic ontological model associated with the aforementioned pair of mixed quantum preparations $\{\rho_{1+2},\rho_{3+4}\}$ for $\theta=0$ saturate the lower bound on the extent of excess ontological distinct inferred from the violation of \eqref{rhoEpistemic}, i.e. this pair is precisely $\frac{1}{\sqrt{2}}\approx 0.707$-distinct.
\end{Observation}
In the Kochen-Spekker model ($\psi$-epistemic), the epistemic states $\{\mu(\lambda|P_{1+2}),\mu(\lambda|P_{3+4})\}$ are ontologically distinguishable with maximum probability,
\begin{widetext}
\beq \label{KSParityOblivious} \nonumber
& s^{1+2,3+4}_\Lambda & = \frac{1}{4\pi}\int_{\Lambda}\max\Bigg\{H(\vec{n_1}\cdot\vec{\lambda})\vec{n_1}\cdot\vec{\lambda}+H(\vec{n_2}\cdot\vec{\lambda})\vec{n_2}\cdot\vec{\lambda},H(\vec{n_3}\cdot\vec{\lambda})\vec{n_3}\cdot\vec{\lambda}+H(\vec{n_4}\cdot\vec{\lambda})\vec{n_4}\cdot\vec{\lambda}\Bigg\}d\lambda \\ \nonumber
& & = \frac{2}{\pi}\int_{\theta_\lambda=0}^{\frac{\pi}{4}}\int_{\Phi_\lambda=0}^{\pi}\cos{\theta_\lambda}\sin^2{\phi_\lambda}d\theta_\lambda d\phi_\lambda \\
& & = \frac{1}{\sqrt{2}},
\eeq
\end{widetext}
where for the second equality expressed the ontic state in standard spherical coordinates $[\lambda_x,\lambda_y,\lambda_z]^T\equiv [\cos{\theta_\lambda}\sin{\phi_\lambda},\sin{\theta_\lambda}\sin{\phi_\lambda},\cos{\phi_\lambda}]^T$ and employed the observations that  $(i)$ $\forall \ \theta_\lambda \in[0,\frac{\pi}{2}]: \ H(\vec{n}_2\cdot\vec{\lambda}))=H(\vec{n}_4\cdot\vec{\lambda}))=0$ and $(ii)$ the Bloch vectors $\vec{n}_1,\vec{n}_2,\vec{n}_3,\vec{n}_4$ are symmetrically distributed on the azimuthal plane which in turn translates to the same contribution from all four quarters of the azimuthal plane. Now, the second equality follows from the observation that, $\frac{1}{4\pi}\int_{\theta_\lambda=0}^{\frac{\pi}{2}}\int_{\Phi_\lambda=0}^{\pi}\max\{\cos{\theta_\lambda}\sin^2{\phi_\lambda},\sin{\theta}\sin^2{\phi}\}d\theta_\lambda d\phi_\lambda=\frac{1}{2\pi}\int_{\theta_\lambda=0}^{\frac{\pi}{4}}\int_{\Phi_\lambda=0}^{\pi}\cos{\theta_\lambda}\sin^2{\phi_\lambda}d\theta_\lambda d\phi_\lambda$. 

\subsection{Measurements}
\begin{figure}   
    \centering
    \includegraphics[scale=0.40]{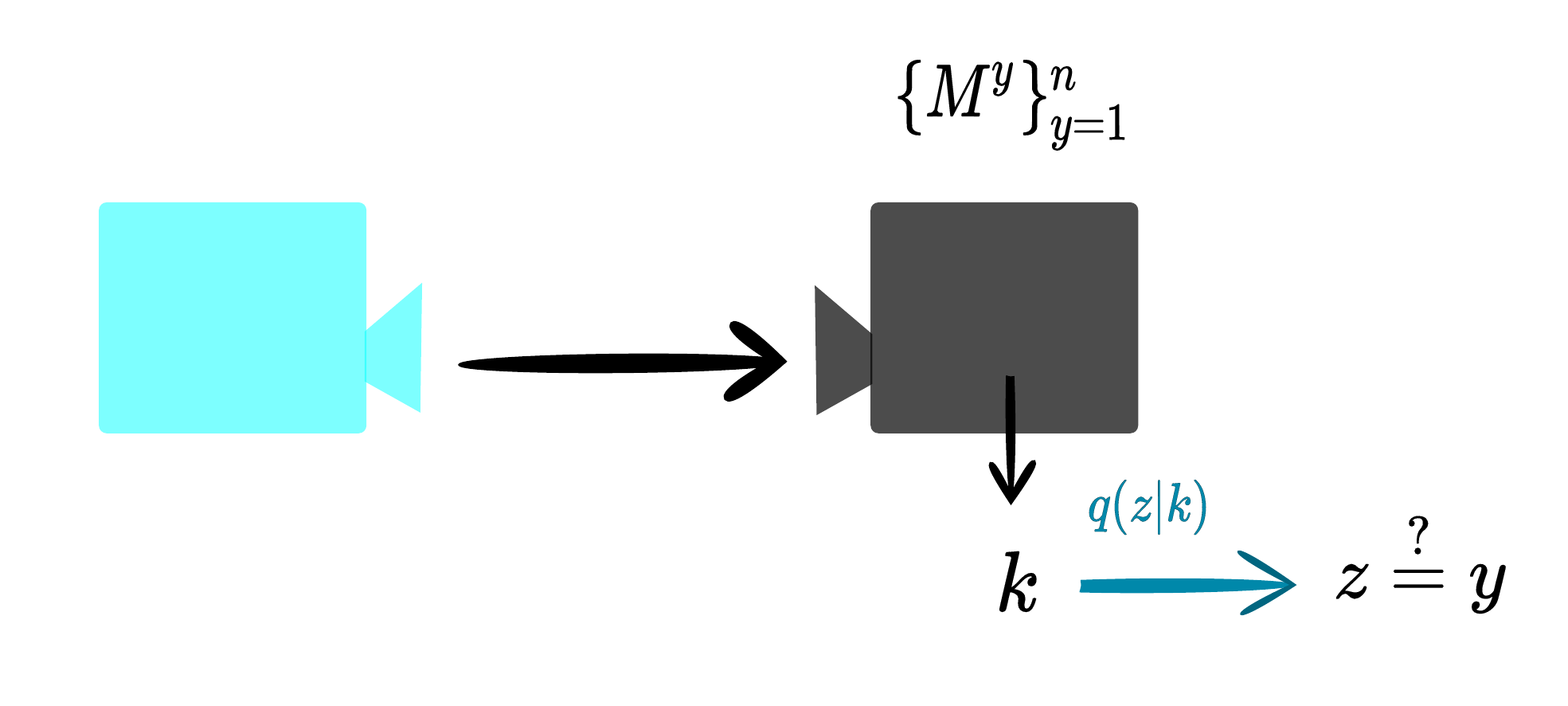}
    \caption{\label{bodm} This figure denotes the set-up for distinguishing a given set of measurements $\{M^y\}^n_{y=1}$ (dark black device) out of a uniform ensemble. In each run of the experiment a measurement $M_y$ is performed on a preparation, the measurement outcome $k$ is then post-processed based on a conditional probability distribution $q(z|k)$ producing the outcome $z$. The (single shot) distinguishing probability in this case, is the probability of producing an output $z=x$ as a result of the $n$-outcome classical post-processing scheme $q(z|k)$ upon a measurement $M^y$ of a preparation.
     This probability is maximized over all preparations (light cyan device) and valid post-processing schemes, yielding the distinguishability $m_\mathcal{O}$ of the constituent measurements. As a consequence of this maximization, $m_\mathcal{O}$ is an intrinsic property of this set of measurements and forms the operational condition accompanying bounded ontological distinctness of measurements.  }
     \label{fig:bodm}
\end{figure}
We can readily extend the principle of bounded ontological distinctness so as to apply to measurements. Here, we invoke sets of $d$-outcome measurements and characterize these sets on the basis of how well the constituent measurements can be operationally distinguished when sampled from a uniform ensemble (FIG. \ref{bodm}). 
\begin{Definition} \textit{$p$-distinguishable measurements:}
In an operational theory, a set of $d$-outcome measurements $\mathcal{M}\equiv \{M^y\}^n_{y=1}$ is called \textit{$p$-distinguishable} if the constituent measurements can be distinguished when sampled from a uniform ensemble with at-most $p$ probability, i.e.,
\beq \label{pDistinguishMeasurements}
& m_{\mathcal{O}} & =  \max_{P} \bigg\{\frac{1}{n} \max_{\{q(z|k)\}}\bigg\{  \sum_{k,y} q(z=y|k)p(k|M^y,P)\bigg\} \bigg\},  \nonumber  \\ 
\nonumber 
& & = p,
\eeq
where $m_{\mathcal{O}}$ denotes the maximum operational probability of distinguishing these measurements when sampled from a uniform ensemble and the first maximization is over the set of all possible operational preparations and the second maximization is over all possible conditional probability distributions $\{q(z|k)\}$, $M^y \in \mathcal{M}$ is a constituent measurement and $k$ is an outcome of the measurement $M^y$ \footnote{Here, we restrict ourselves to considering simple single-shot measurement distinguishability without the aid of more exotic features of operational theories such as entanglement \cite{RevModPhys.81.865,PhysRevA.90.052312}. } (FIG. \ref{fig:bodm}). 
\end{Definition}

Clearly, the initial maximization relieves $m_{\mathcal{O}}$ of its dependence on preparations and the second maximization relieves $m_{\mathcal{O}}$ of its dependence on classical post-processing schemes, deeming it to be a suitable characterizing feature of sets of measurements. We call the probability distributions $\{q(z|k)\}$  classical post-processing schemes which satisfy: \textit{positivity:} $\forall \ k,z: \ q(k|z) \geqslant 0$ and, \textit{completeness:} $\forall \ z: \sum_{k}q(k|z)=1$.  Observe that the set of all possible $n$-outcome classical post-processing schemes forms a convex polytope, the extremal points of this polytope are deterministic probability distributions, i.e., each extremal classical post-processing scheme is of the form: for each $k$, $q(z|k)=0$ except for specific $z=z_k$ for which $q(z_k|k)=1$. In light of this observation, we can readily solve the second maximization by finding out the optimal classical post-processing scheme. Clearly, for each preparation $P$ and measurement outcome $k$, the optimal classical post-processing would be to output the index $y$ of the measurement which assigns the largest probability to $k$. This in-turn leads us to the following succinct expression for $m_\mathcal{O}$, 
\beq \label{pDistinguishM}
m_\mathcal{O}= \max_{P} \Bigg\{\frac{1}{n}  \sum_{k}\max_y \Bigg\{p(k|M^y,P)\Bigg\}\Bigg\}.
\eeq
Recall that in an ontological model, each set of operational measurements is associated with a set of response schemes. We characterize the sets of response schemes based on how distinct the constituent response schemes are. 
\begin{Definition} \textit{$p$-distinct response schemes:}
In an ontological model, a set of $d$-outcome response schemes $\mathcal{M}_{\Lambda}$ is called $p$-distinct if the constituent response schemes can be ontological distinguished (upon having fine-grained control over preparation of the ontic state $\lambda$) when sampled from a uniform ensemble with at most $p$ probability, i.e., 
\be \label{pLambdaDistinguishM}
m_{\Lambda}=\max_{\lambda}\Bigg\{\frac{1}{n} \sum_{k} \max_{y} \Bigg\{\xi(k|\lambda,M^y)\Bigg\} \Bigg\}=p,
\ee
where $m_{\Lambda}$ denotes the maximum probability of distinguishing these response schemes when sampled from a uniform ensemble given the ontic state $\lambda$ and the maximization is over all ontic states $\lambda$. 
\end{Definition}
Now as the operational preparation might not allow fine-grained control on the ontic states, we have in general $m_\mathcal{O}\leqslant m_\lambda$. \\
Equipped with these measures we present our ontological principle for measurements: \\
\begin{Definition}
\textit{Bounded ontological distinctness for measurements $(BOD_M)$}: The set of response schemes $\mathcal{M}_\Lambda \equiv \{\{\xi(k|\lambda,M^y)\}\}^n_{y=1}$ underlying a set of $p$-distinguishable measurements $\mathcal{M}\equiv \{M^y\}^n_{y=1}$ is $p$-distinct, i.e., $m_{\Lambda} = m_{\mathcal{O}}=p$, yielding the ontological constraint,
\beq \label{BODM} 
\max_{\lambda}\Bigg\{\frac{1}{n} \sum_{k} \max_{y} \Bigg\{\xi(k|\lambda,M^y)\Bigg\} \Bigg\}=p.
\eeq
\end{Definition}
When formulated in this way, $BOD_M$ serves as a criterion for characterization of ontological models of a given operational theory. It is useful, at this point, to define a criterion for characterization of operational theories, namely, an operational theory or a fragment thereof is said to satisfy $BOD_M$, if there exists an ontological model which satisfies $BOD_M$ for all sets of prescribed operational measurements. Conversely, an operational theory or a fragment thereof is said to violate $BOD_M$ if there exists no ontological model which satisfies $BOD_M$ for all sets of prescribed operational measurements. \\
Unlike preparations, the Beltrametti-Bugajski $\psi$-complete ontological model adheres to $BOD_M$, implying that operational quantum theory satisfies $BOD_M$. This follows from the observations $(i)$ in this model the response scheme underlying a POVM element $M_k$ is simply $\xi(k|\lambda,M)=\Tr(\ketbra{\psi}{\psi}M_k)$, and $(ii)$ we have fine-grained control over the ontic state $\lambda$ which is simply the quantum state $\lambda\equiv \ket{\psi}$. This line of reasoning parallels the one in \cite{spekkens2005contextuality} for the impossibility of retrieving a quantum violation of measurement noncontextuality without additional assumptions.\\
However, we find that the ontological models wherein quantum projective measurements are associated with deterministic response schemes violate $BOD_M$. In-order to substantiate this, we make the following formal assumption:\\
\begin{Definition}
\textit{Outcome deterministic response schemes:} In an ontological model each operational (extremal) measurement $M$ is associated with a deterministic response scheme $\{\xi(z|\lambda,M)\}$ such that $\forall \ z, \lambda,M: \  \xi(z|\lambda,M) \in \{0,1\}$.  \\
\end{Definition}

Consider a pair of two-outcome of measurements $\{M^1,M^2\}$, then in an outcome deterministic ontological model,
\be \label{BODMPlusOD}
m_{\Lambda}=\frac{1}{2}\max_{\lambda}\Bigg\{ \sum_{k\in\{1,2\}} \max\Bigg\{\xi(k|\lambda,M^1),\xi(k|\lambda,M^2)\Bigg\} \Bigg\}
\ee
belongs to $\{\frac{1}{2},1\}.$
If the pair of measurements $\{M^1,M^2\}$ are neither completely indistinguishable nor perfectly distinguishable, i.e., $m_\mathcal{O}\in (\frac{1}{2},1)$, then the associated deterministic response schemes must be perfectly distinguishable $m_\Lambda=1$. For instance, the Kochen-Specker model prescribes outcome deterministic response schemes to two-outcome two-dimensional projective quantum measurements $\xi(1|\lambda,\{\ketbra{\phi}{\phi},\ketbra{\phi^\perp}{\phi^\perp}\})=H(\vec{n}_\phi \cdot\vec{\lambda})$ and $\xi(2|\lambda,\{\ketbra{\phi}{\phi},\ketbra{\phi^\perp}{\phi^\perp}\})=H(\vec{n}_{\phi^\perp}\cdot\vec{\lambda})$. These response schemes return outcome $1$ whenever the ontic-state lies in the Bloch hemisphere with central axis $\vec{n_\phi}$ and outcome $2$ otherwise. Now, consider a pair of two dimensional projective measurements $\{M^1_1\equiv \ketbra{\phi_1}{\phi_1},M^1_2\equiv \ketbra{\phi^\perp_1}{\phi^\perp_1}\}$ and $\{M^2_1\equiv \ketbra{\phi_2}{\phi_2},M^2_2\equiv \ketbra{\phi^\perp_2}{\phi^\perp_2}\}$. If $\phi_1\neq \phi_2$ then one can always find suitable ontic state $\vec{\lambda}$ such that $\xi(1|\lambda,\{\ketbra{\phi_1}{\phi_1},\ketbra{\phi_1^\perp}{\phi_1^\perp}\})=1$ and $\xi(1|\lambda,\{\ketbra{\phi_2}{\phi_2},\ketbra{\phi_2^\perp}{\phi_2^\perp}\})=0$ and similarly for the second outcome, deeming the two response schemes under consideration to be perfectly distinguishable with $m_\Lambda=1$.

\subsection{Transformations}
 \begin{figure}     
    \centering
    \includegraphics[scale=0.35]{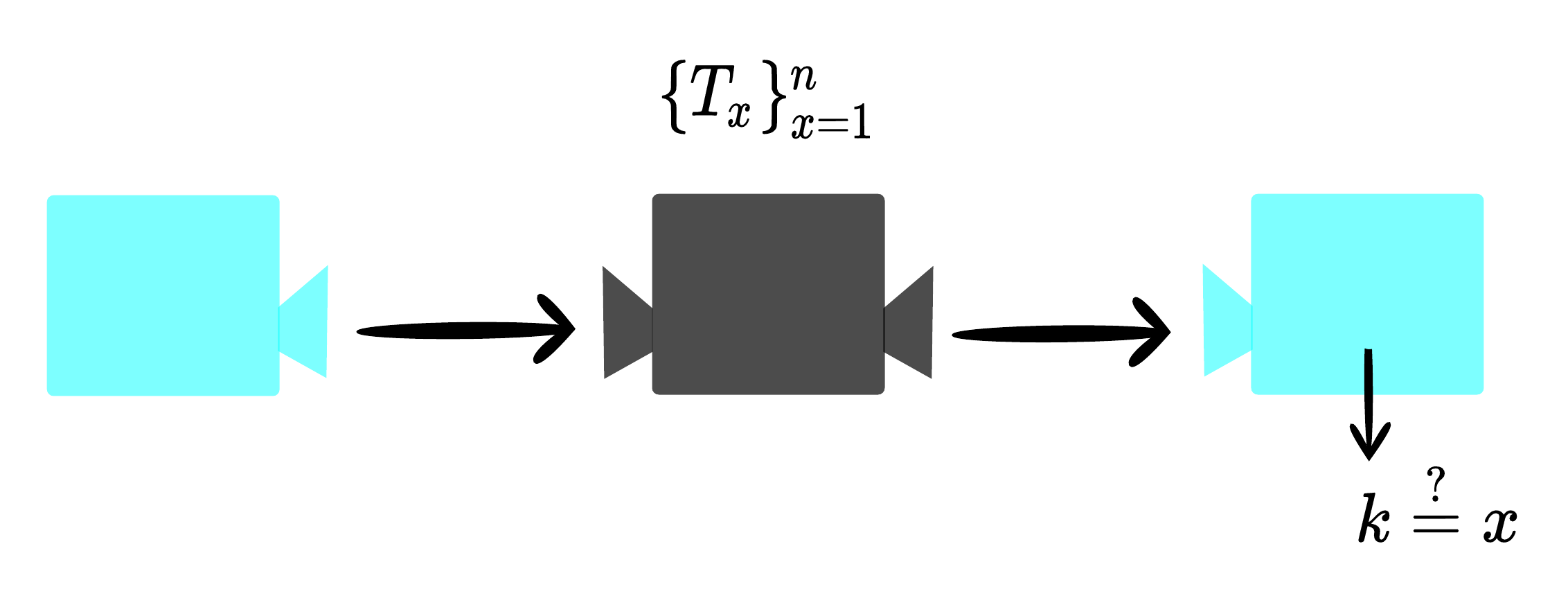}
    \caption{
    This figure denotes the set-up for distinguishing a given set of transformations $\{T_x\}^n_{x=1}$ (dark black device) out of a uniform ensemble. In each run of the experiment, a preparation undergoes an transformation $T_x$, followed by a measurement producing the outcome $k$. The (single shot) distinguishing probability in this case, is the probability of producing an output $k=x$ as a result of $n$-outcome measurement.
     This probability is maximized over all preparations and measurements (light cyan devices), yielding the distinguishability $t_\mathcal{O}$ of the constituent measurements. As a consequence of this maximization $t_\mathcal{O}$ is an intrinsic property of this set of transformations and forms the operational condition accompanying bounded ontological distinctness of transformations.}
     \label{fig:bodt}
\end{figure}
 
 Now, we invoke sets of transformations. We characterize these sets based on how well the constituent transformations can be operationally distinguished (FIG. \ref{fig:bodt}). 
\begin{Definition} \textit{$p$-distinguishable transformations:}
 In an operational theory, a set of transformations $\mathcal{T}\equiv \{T_x\}^n_{x=1}$ is called \textit{$p$-distinguishable} if the constituent transformations can be distinguished when sampled from a uniform ensemble with at-most $p$ probability, i.e.,
\be \label{pDistinguishT}
t_{\mathcal{O}}= \max_{P}\max_{M}\Bigg\{\frac{1}{n} \sum_{x}p(k=x|P,T_x,M)\Bigg\}= p,
\ee
where $t_{\mathcal{O}}$ denotes the maximum operational probability of distinguishing these transformations when sampled from a uniform ensemble, the initial maximization is over all possible preparations and the second maximization is over the set of all possible $n$-outcome  measurements in a given operational theory \footnote{Yet again we use the simplest single-shot prepare, transform and measure set-up for transformation distinguishability, without employing more intricate features of operational theories such as entanglement \cite{PhysRevLett.114.060404}.} (FIG. \ref{fig:bodt}). 
\end{Definition}
Observe that the maximum probability of distinguishing the transformations $t_{\mathcal{O}}$, is maximized over all possible preparations and measurements available in an operational theory. This maximization relieves $t_{\mathcal{O}}$ of its dependence on preparations and measurements, deeming it to be a suitable characterizing feature of sets of transformations. \\

Each set of operational transformations $\mathcal{T}$ is associated with a set of transition schemes $\mathcal{T}_\Lambda$. We characterize these sets based on how distinct the constituent transition schemes actually are. \\
\begin{Definition}
 \textit{$p$-distinct transition schemes:}
In an ontological model, a set of transition schemes $\mathcal{T}_\Lambda \equiv \{\{\gamma(\lambda'|\lambda,T_x)\}\}^n_{x=1}$ is termed $p$-distinct if the constituent transition schemes can be ontologically distinguished (upon having fine-grained control over preparation of the ontic state $\lambda$ and access to the post-transition ontic state $\lambda'$) when sampled from a uniform ensemble with at-most $p$ probability, i.e.,
\beq \label{pLambdaDistinguishT}
t_{\Lambda} &=& \max_{\lambda}\Bigg\{\max_{\{\xi(k|\lambda')\}}\Bigg\{\frac{1}{n} \sum_{x}\int_{\Lambda}\gamma(\lambda'|\lambda,T_x)\xi(k=x|\lambda')d\lambda' \Bigg\}\Bigg\} \nonumber  \\
&=& p ,
\eeq
where $t_{\Lambda}$ denotes the maximum probability of distinguishing these transition schemes when sampled from a uniform ensemble given access to initial ontic state $\lambda$, $\Lambda$ represents the ontic state space, $\{\gamma(\lambda'|\lambda,T_x)\}$ is the transition scheme associated the transformation $T_x$ and the maximization is over all ontic states $\lambda$ as well as all valid response schemes $\{\xi(k|\lambda)\}$ which satisfy \textit{positivity} and \textit{completeness}. 
\end{Definition}
Yet again, following from the observation that the response schemes are so constrained form a convex polytope, we can readily solve the second maximization which yields the following succinct expression of ontological distinctness,
\be \label{pLambdaDistinguishT2}
t_{\Lambda} = \max_{\lambda}\Bigg\{\frac{1}{n} \int_{\Lambda}\max_{x}\Bigg\{\gamma(\lambda'|\lambda,T_x)\Bigg\}d\lambda' \Bigg\}.
\ee \\
Now we present our ontological principles for transformation: \\
\begin{Definition}
\textit{Bounded ontological distinctness for transformations ($BOD_T$):} The set of transition schemes $\mathcal{T}_\Lambda \equiv \{\{\gamma(\lambda'|\lambda,T_x)\}\}^n_{x=1}$ underlying a set of $p$-distinguishable transformations $\mathcal{T}\equiv \{T_x\}^n_{x=1}$ is $p$-distinct, i.e., $t_\Lambda=t_{\mathcal{O}}=p$, yielding the ontological constraint,
\be \label{BODT}
\max_{\lambda}\Bigg\{\frac{1}{n} \int_{\Lambda}\max_{x}\Bigg\{\gamma(\lambda'|\lambda,T_x)\Bigg\}d\lambda' \Bigg\}=p.
\ee
\end{Definition}
When formulated in this way, $BOD_T$ serves as a criterion for characterization of ontological models of a given operational theory. It is useful, at this point, to define a criterion for characterization of operational theories, namely, an operational theory or a fragment thereof is said to satisfy $BOD_T$, if there exists an ontological model which satisfies $BOD_T$ for all sets of prescribed operational transformations. Conversely, an operational theory or a fragment thereof is said to violate $BOD_T$ if there exists no ontological model which satisfies $BOD_T$ for all sets of prescribed operational transformations. \\
Yet again, in this form \eqref{BODT}, $BOD_T$ leads to some interesting consequences. We present one such consequence in the form of the following proposition.
\subsubsection{Bounded ontological distinctness of three transformations}
\begin{Proposition} \label{propositionBODT}
If an operational theory admits ontological models adhering to $BOD_T$, then a set of three $p$-distinguishable transformations $\mathcal{T}\equiv\{T_1,T_2,T_3\}$, are pairwise distinguishable with respect to the same initial preparation with the average probability of success being at-most $\frac{1+p}{2}$ , i.e.,
\be \label{BODineqT}
\frac{1}{3}\Bigg(\max_{P}\Bigg\{t^{1,2}_\mathcal{O}(P)+t^{2,3}_\mathcal{O}(P)+t^{3,1}_\mathcal{O}(P)\Bigg\}\Bigg)\leqslant \frac{1+p}{2},
\ee
where $t^{i,j}_{\mathcal{O}}(P)$ denotes the maximum probability of distinguishing the transformations $\{T_i,T_j\}$ out of a uniform ensemble given the initial preparation $P$, i.e., $t^{i,j}_\mathcal{O}(P)=\max_{M}\{\frac{1}{2} \sum_{x\in \{i,j\}}p(k=x|P,T_x,M)\}$.
\end{Proposition}
\begin{proof}
We follow roughly the same steps as in the proof of Proposition \ref{propositionBOD}. The maximum probability of distinguishing a set of three transition schemes $\mathcal{T}_\Lambda\equiv\{\{\gamma(\lambda'|T_1,\lambda)\},\{\gamma(\lambda'|T_2,\lambda)\},\{\gamma(\lambda'|T_3,\lambda)\}\}$ upon having fine grained control over the initial ontic state $\lambda$ as well as complete access to post-transition ontic state $\lambda'$ has the expression,
\begin{widetext}
\beq \label{pLambdaDistinguish3T}\nonumber
&t_{\Lambda} & =  \frac{1}{3}\max_{\lambda}\Bigg\{\int_{\Lambda}\max_{x}\Bigg\{\gamma(\lambda'|\lambda,T_x)\Bigg\}d\lambda' \Bigg\} \\ 
& & = 1-\frac{1}{3} \min_\lambda\int_\Lambda\Bigg\{\min\bigg\{\gamma(\lambda'|T_1,\lambda)+\gamma(\lambda'|T_2,\lambda),\gamma(\lambda'|T_2,\lambda)+\gamma(\lambda'|T_3,\lambda),\gamma(\lambda'|T_3,\lambda)+\gamma(\lambda'|T_1,\lambda)\bigg\}
\Bigg\}d\lambda',
\eeq
Similarly, the expression for maximal average pairwise ontological distinctness with respect to same initial ontic state is,
\beq \label{preBODineqT1} \nonumber
&& \frac{1}{3}\max_\lambda\bigg\{t^{1,2}_{\Lambda}(\lambda)+t^{2,3}_{\Lambda}(\lambda)+t^{3,1}_{\Lambda}(\lambda)\bigg\} \\
&=&  1- \frac{1}{6}\min_\lambda \Bigg\{\int_\Lambda \Bigg( \min\bigg\{\gamma(\lambda'|\lambda,T_1),\gamma(\lambda'|\lambda,T_2)\bigg\} + \min\bigg\{\gamma(\lambda'|\lambda,T_2),\gamma(\lambda'|\lambda,T_3)\bigg\}+\min\bigg\{\gamma(\lambda'|\lambda,T_3),\gamma(\lambda'|\lambda,T_1)\bigg\}\Bigg) d\lambda' \Bigg\} \nonumber \\ \nonumber
& \leqslant&  1- \frac{1}{6}\min_\lambda\int_\Lambda\Bigg\{\min\bigg\{\gamma(\lambda'|\lambda,T_1)+\gamma(\lambda'|\lambda,T_2),\gamma(\lambda'|\lambda,T_2)+\gamma(\lambda'|\lambda,T_3), \gamma(\lambda'|\lambda,T_3)+\gamma(\lambda'|\lambda,T_1)\bigg\}\Bigg\}d\lambda' \\
&=& \frac{1+t_\Lambda}{2}.
\eeq
\end{widetext}
As operational preparations might not allow fine-grained control over the initial ontic-state and operational measurements might not reveal the post-transition ontic state completely, we have in general that $\frac{1}{3}(\max_{P}\{t^{1,2}_\mathcal{O}(P)+t^{2,3}_\mathcal{O}(P)+t^{3,1}_\mathcal{O}(P)\})\leqslant \frac{1}{3}\max_\lambda\{t^{1,2}_{\Lambda}(\lambda)+t^{2,3}_{\Lambda}(\lambda)+t^{3,1}_{\Lambda}(\lambda)\}$. This along with \eqref{preBODineqT1} yields,
\be \nonumber \label{preBODineqT2}
\frac{1}{3}\Bigg(\max_{P}\Bigg\{t^{1,2}_\mathcal{O}(P)+t^{2,3}_\mathcal{O}(P)+t^{3,1}_\mathcal{O}(P)\Bigg\}\Bigg)  \leqslant \frac{1+t_\Lambda}{2}.
\ee
Finally,  the desired thesis follows from $BOD_T$ which implies $t_\Lambda=t_\mathcal{O}=p$. 
\end{proof}
\subsubsection{Excess ontological distinctness of three quantum transformations} \label{markme3}
In order to demonstrate quantum violation of Proposition \ref{propositionBODT}, we consider three qubit unitary transformations $U_1\equiv \mathbb{I}$,$U_2\equiv \mathcal{R}(\frac{2\pi}{3})$ and $U_3\equiv \mathcal{R}(\frac{4\pi}{3})$ where $\mathcal{R}(\theta)$ symbolizes a rotation of angle $\theta$ in the x-z plane of the Bloch sphere. It is easy to see that for any initial pure state preparation the post transformations states and consequently the transformations themselves are $\frac{2}{3}\approx 0.667$-distinguishable i.e. $t_Q=\frac{2}{3}$. Furthermore, for any initial pure state preparation, the post transformation states are pair-wise distinguishability with average probability of success being $\frac{1}{3}(\max_{P}\{t^{1,2}_\mathcal{O}(P)+t^{2,3}_\mathcal{O}(P)+t^{3,1}_\mathcal{O}(P)\})=\frac{1}{2}(1+ \frac{\sqrt{3}}{2})\approx 0.933$, which is a violation of Proposition \ref{propositionBODT}, thereby demonstrating excess ontological distinctness of quantum transformations.

\section{Excess ontological distinctness powers quantum advantage in communication tasks}
In this section, we lay down a general framework for communication tasks fuelled by excess ontological distinctness. We consider bipartite one-way communication tasks wherein Alice receives an input $x$ based on a prior probability distribution $p(x)$, similarly Bob receives an input $y$ according to a prior probability distribution $p(y)$ and produces an output $z$. They repeat the task for several rounds so as to gather frequency statistics in the form of conditional probabilities $p(z|x,y)$. Their goal is to maximize an associated success metric (figure of merit) $succ_{\mathcal{O}}$ which has the following generic form, 
\be \label{succMetric}
succ_\mathcal{O} = \sum_{x,y,z} c(x,y,z)p(z|x,y),
\ee
where the coefficients $c(x,y,z)$ are positive reals. Observe that, as the success metric is a linear function of the conditional probabilities $p(z|x,y)$, the prior probability distributions $p(x),p(y)$ are of no consequence, and we can take the inputs $x,y$ to be uniformly distributed, without loss of generality.
In general, Alice and Bob can employ various strategies to maximize the figure of merit, however, here we only consider the simplest prepare and measure protocols, wherein Alice's encodes her input onto an operational preparation $P_x$ and Bob performs an operational measurement $M^y$ to produce an outcome $z$ such that the operational predictions and resulting statistics are of the form $p(z|x,y)=p(z|P_x,M^y)$. Unlike the communication complexity problems, in these tasks we do not restrict the dimension of the communicated system, instead we constrain the distinguishability of Alice's preparations. Below we describe two different classes of communication tasks characterized by distinct distinguishability constraints (see FIG. \ref{ComTask}).
\begin{figure}     
    \centering
    \includegraphics[width=\linewidth]{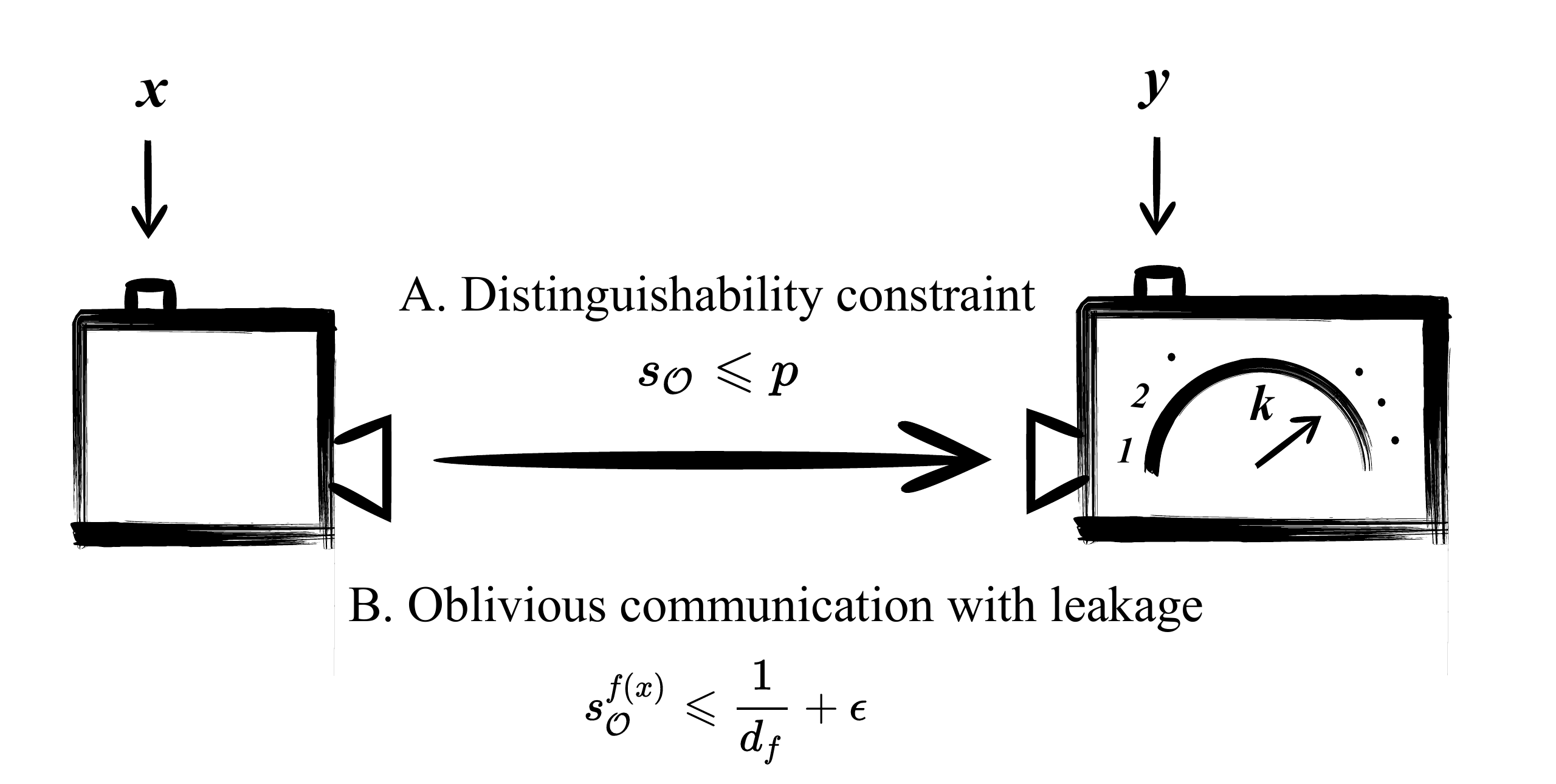}
    \caption{\label{ComTask}
    The figure displays the set-up of generic bipartite communication tasks. In each run of the task, the sender (Alice) encodes her input $x$ onto a preparation and transmits it to the receiver (Bob). Bob performs a measurement based on his input $y$ and produces and output $k$. Their aim to maximize a success metric $succ_\mathcal{O}$ associated with the task which has the form \eqref{succMetric}. There is no constraint on the amount of communication, however, we constrain the communication in two ways: A. bounding the distinguishability of the sender's preparations, and B. bounding the distinguishability of sender's preparations associated with different values of a function (oblivious function) of the sender's input $f(x)$. }
    \label{fig:ct}
\end{figure}

\subsection{Constraining the distinguishability of sender's preparations}
 First, we bound the probability of distinguishing Alice preparations out of a uniform ensemble. Consequently, the communication channel is constrained such that Alice may only transmit information encoded on $p$-distinguishable preparations \eqref{pDistinguish}. This in-turn imposes non-trivial constrains on the performance of different operational theories in communication tasks.\\
 For instance, consider a task wherein Alice receives a uniformly distributed trit $x\in\{1,2,3\}$ as her input, while Bob receives a uniformly distributed bit $y\in\{1,2\}$ as his input produces an output bit $z\in\{1,2\}$. Their aim to maximize the associated figure of merit \eqref{succMetric} with coefficients,
\be \label{commTask1}
c(x,y,z)= \begin{dcases}
1 \ & \text{if } x\in\{1,2\},y=1,z=1, \\
1 \ & \text{if } x=3,y=1,z=2, \\
1 \ & \text{if } x=1,y=2,z=1, \\
1 \ & \text{if } x=2,y=2,z=2, \\
0 \ & \text{otherwise.}
\end{dcases}
\ee
If the involved parties employ resources pertaining to theories that satisfy $BOD_P$, the operational $p$-distinguishability constraint yields the following upper bound on the success metric \eqref{commTask1} of this communication task,
\begin{Proposition} \label{propositionBOD2} If an operational theory admits ontological models adhering to the principle of \textit{bounded ontological distinctness}, then for a set of three at-most $p$-distinguishable preparations $\mathcal{P}\equiv\{P_1,P_2,P_3\}$, the value of the success metric \eqref{commTask1} is upper bounded as follows,
\beq \label{BODineq2} \nonumber
 succ_\mathcal{O} & =&  p(z=1|x=1,y=1)+p(z=1|x=1,y=2) \nonumber \\
 && +p(z=1|x=2,y=1)  + p(z=2|x=2,y=2) \nonumber \\ \nonumber
 && +p(z=2|x=3,y=1) \\
& \leqslant & 2+3p.
\eeq   
\end{Proposition}
\begin{proof}
In an ontological model the maximum value of the operational success metric has the expression,
\begin{widetext}
\beq \nonumber
& succ_\mathcal{O}= &\max_M\Bigg\{\int_\Lambda \bigg(\mu(\lambda|P_1)+\mu(\lambda|P_2)\bigg)\xi(k=1|\lambda,M)+\mu(\lambda|P_3)\xi(k=2|\lambda,M) d\lambda\Bigg\} \\
& & + \max_M\Bigg\{\int_\Lambda \mu(\lambda|P_1)\xi(k=1|\lambda,M)+\mu(\lambda|P_2)\xi(k=2|\lambda,M) d\lambda\Bigg\}.
\eeq
This equation facilitates the alternative description of communication task under consideration as a distinguishability task, wherein Bob upon receiving $y=1$ tries to discern the preparations $\{P_1,P_2\}$ from $P_3$, and upon receiving $y=2$ attempts to distinguish $P_1$ from $P_2$. Consequently, upon having access to the ontic state $\lambda$ and in light of the observations regarding response schemes employed before \eqref{pLambdaDistinguishSuccinct} the maximum value of the ontological success metric has the expression,
\beq \label{BODineq2support}
\nonumber
&succ_\Lambda &=\int_\Lambda \Bigg( \max\bigg\{\mu(\lambda|P_1)+\mu(\lambda|P_2),\mu(\lambda|P_3)\bigg\}+\max\bigg\{\mu(\lambda|P_1),\mu(\lambda|P_2)\bigg\}\Bigg) d\lambda, \\
&            &=5-\int_\Lambda \Bigg( \min\bigg\{\mu(\lambda|P_1)+\mu(\lambda|P_2),\mu(\lambda|P_3)\bigg\}+\min\bigg\{\mu(\lambda|P_1),\mu(\lambda|P_2)\bigg\}\Bigg) d\lambda,
\eeq
where the second equality follows from the fact $\forall \ (a\geqslant0,b\geqslant0): \max\{a,b\}=a+b-\min\{a,b\}$. \\
Moving on, we invoke \eqref{pLambdaDistinguish3}, which yields,
\beq \label{BODineq2support1} 
&\int_\Lambda \min\bigg\{\mu(\lambda|P_1)+\mu(\lambda|P_2),\mu(\lambda|P_2)+\mu(\lambda|P_3),\mu(\lambda|P_3)+\mu(\lambda|P_1)\bigg\}
d\lambda &\geqslant 3(1-s_\Lambda)
\eeq
\end{widetext}
where $s_\Lambda$ is the ontological distinctness for the three preparations consideration. Observe that as the operational measurements might not reveal the ontic-state $\lambda$ precisely, we have in general $succ_{\Lambda}\geqslant succ_{\mathcal{O}}$. Now employing the fact that $\forall \ (a\geqslant0,b\geqslant0,c\geqslant0): \min\{a+b,c\}+\min\{a,b\}\geqslant \min\{a+b,b+c,c+a\}$, and plugging \eqref{BODineq2support1} into \eqref{BODineq2support} we arrive at,
\beq
succ_{\mathcal{O}}\leqslant 2+3s_\Lambda.
\eeq
The desired thesis follows from $BOD_P$ which implies $s_\Lambda=s_{\mathcal{O}}=p$. 
\end{proof}  
Moreover, even the average pair-wise distinguishability of three preparations considered in Proposition \ref{propositionBOD} can be cast as the success metric of such a communication task wherein Alice and Bob receive trits as inputs $x,y\in\{1,2,3\}$, Bob outputs a bit $z\in\{1,2\}$ and the success metric has the following coefficients:
\be \label{comtask2}
c(x,y,z)=
\begin{dcases}
\frac{1}{6} \ & \text{if } z=x,x\in\{1,2\},y=1 \\
\frac{1}{6} \ & \text{if } z=x,x\in\{2,3\},y=2, \\
\frac{1}{6} \ & \text{if } z=x,x\in\{3,1\},y=3, \\
0 \ & \text{otherwise.}
\end{dcases}
\ee
Consequently, Proposition \ref{propositionBOD} yields a bound on the success metric of this task $succ_{\mathcal{O}}\leq \frac{1+p}{2}$ for theories that satisfy $BOD_P$.

Operationally, Alice and Bob can employ classical or quantum resources and protocols to aid them in these tasks. In general, one can define operational classical communication models as,
\begin{Definition} \label{Classicality} A \textit{classical communication models} for a prepare and measure scenario, is one wherein there exists a function $\Omega(\lambda)=\omega$ such that,
\begin{enumerate}
    \item $\omega$ is a sufficient operational statistic for the ontic state $\lambda$, i.e. $\forall \ \lambda,k,M \ : \ p(k|\omega=\Omega(\lambda),M)=p(k|\lambda,M)$, 
    \item there exists operational measurements which allow complete access to $\omega$.
\end{enumerate}
\end{Definition}
A relevant example of such models is $d$-levelled classical communication where $\omega\in\{1,\ldots,d\}$. A $d$-levelled classical communication protocol entails Alice encoding her input $x\in\{1,\ldots,n\}$ onto a classical message $\omega\in\{1,\ldots,d\}$ based on an encoding scheme i.e. condition probability distributions of the form $p_{\mathcal{E}}(\omega|x)$, Bob upon receiving the message employs a decoding scheme i.e. conditional probability distributions of the form $\{p_{\mathcal{D}}(z|\omega,y)\}$ to produce the output $z$. The operational predictions can be summarized as, 
\be \label{opPredictions}
p(z|x,y)=\sum_\omega p_{\mathcal{E}}(\omega|x)p_{\mathcal{D}}(z|\omega,y).
\ee
The operational $p$-distinguishability condition for preparations \eqref{pDistinguish} translates to the following $p$-distinguishability condition for encoding schemes,
\beq \label{classOpCond}
s_C=\frac{1}{n}\sum_\omega \max_x\bigg\{p_{\mathcal{E}}(\omega|x)\bigg\}\leqslant p. 
\eeq
In general, classical communication models and in particular for $d$-levelled classical communication satisfy $BOD_P$, as one can readily take $\omega$ as the ontic state, and consequently, the operational $p$-distinguishability conditions \eqref{pDistinguish} and \eqref{classOpCond} imply the ontological $p$-distinctness constraint \eqref{pLambdaDistinguish}. This argument substantiates the candidature of bounded ontological distinctness as a notion of classicality. Furthermore, this in-turn implies that for classical communication models including $d$-levelled classical communication constrained by \eqref{classOpCond}, the maximal value of the success metric is the same as for operational theories that satisfy $BOD_P$, i.e., $succ_C\leqslant 2+3p$ and $succ_C\leqslant \frac{1+p}{2}$ for the aforementioned communication tasks specified by \eqref{commTask1} and \eqref{comtask2}. At point this it is important to note that we allow Alice and Bob unrestricted access to pre-shared randomness, however we observe that,
\begin{Observation}
Shared randomness does not yield any additional advantage to the performance of $d$-levelled classical communication strategies in the communication tasks constrained by the operational condition \eqref{classOpCond} (evaluated in presence of shared randomness).
\end{Observation}
\begin{proof}
Let Alice and Bob share an arbitrary discrete random variable $r$ with an associated probability distribution $p(r)$ before the communication task begins. In this case, Alice encodes her input $x$ onto the $d$-levelled message $\omega$ based on the shared random variable $r$, according to the encoding scheme $p_{\mathcal{E}}(\omega|x,r)$. Similarly, upon receiving $\omega$, Bob produces his output $z$ based on $r$, according to a decoding scheme $p_{\mathcal{D}}(z|\omega,r,y)$. In this case, the operational predictions can be summarised as $p(z|x,y)=\sum_{\omega,r} p(r) p_{\mathcal{E}}(\omega|x,r)p_{\mathcal{D}}(z|\omega,r,y)$. Since there is no restriction on the dimension of the classical message, we can absorb the shared random variable into the classical message such that $\omega'=(\omega,r)$. As the pre-shared randomness is uncorrelated with Alice's input, we have $p(r)p_{\mathcal{E}}(\omega|x,r)=p_{\mathcal{E}}(\omega,r|x)$, which leads us to the expression for operational predictions $p(z|x,y)=\sum_{\omega'} p_{\mathcal{E}}(\omega'|x)p_{\mathcal{D}}(z|\omega',y)$, which is identical to \eqref{opPredictions}. 
In light of this modification, the operational condition \eqref{classOpCond} as well the operational consequences such as $succ_C\leqslant 2+3p$ and $succ_C\leqslant \frac{1+p}{2}$ remain unaltered, where $p$ is maximum probability of distinguishing the preparations in presence of shared randomness. 
\end{proof}
\label{markmeup2}
Alternatively, Alice and Bob can employ a quantum prepare and measure protocol to aid them in these communication tasks. A prepare and measure quantum protocol entails preparation and transmission of Alice's quantum state $\rho_x$ for her input $x$, followed by a measurement at Bob's end $\{M^y_z\}$ based on his input $y$. The operational predictions can be summarized as $p(z|x,y)=\Tr(\rho_xM^y_z)$. The operational $p$-distinguishability condition \eqref{pDistinguish} translates to the constraint \eqref{pQuantumDistinguish} for quantum communication. For the simple communication task described above consider the quantum protocol wherein Alice prepares the state $\rho_x=\frac{\mathbb{I}+\vec{n}_x\cdot \vec{\sigma}}{2}$ based on her input $x$, where $\vec{n_1}=[1,0,0]^T$, $\vec{n_2}=[0,1,0]^T$ and $\vec{n_3}=\frac{-\vec{n_1}-\vec{n_2}}{\sqrt{2}}$ (see FIG. \ref{PrepMeas4}). These preparations are $\frac{2}{3}\approx 0.667$-distinguishable, consequently in an operational theory which satisfies $BOD_P$ the success metric would be bounded as $succ_O \leqslant 4$. On the other hand, Bob's measures $M^y=\vec{s_y}\cdot \vec{\sigma}$ where $\vec{s_1}=\frac{\vec{n_1}+\vec{n_2}}{\sqrt{2}}$ and $\vec{s_2}=\frac{\vec{n_1}-\vec{n_2}}{\sqrt{2}}$. This protocol achieves $succ_Q=3+\sqrt{2}\approx 4.414$ which violates Proposition \ref{propositionBOD2}, thereby demonstrating excess ontological distinctness of quantum preparations and yielding an advantage over classical communication. As a consequence, the epistemic states underlying this set of $\frac{2}{3}\approx 0.667$-distinguishable preparations must be at-least $\frac{1+\sqrt{2}}{3}\approx 0.805$-distinct. Similarly, the quantum states and measurements that violate Proposition \ref{propositionBOD} can be cast as advantageous quantum prepare and measure protocols \footnote{For a given instance of the communication task, if the operational condition, i.e the maximum distinguishability (say $p$) of the preparations is pre-determined, Alice can employ any set of preparations $\{P_x\}$ as long as they are at-most $p$-distinguishable. On the other hand, for any set of preparations $\{P_x\}$ one can find the appropriate operational condition, i.e. their distinguishability. A subsequent violation of Proposition \ref{propositionBOD} or \ref{propositionBOD2} would then imply a quantum advantage in the associated communication task over all similarly constrained classical encoding schemes. Therefore, any instance of violation of Proposition \ref{propositionBOD} or \ref{propositionBOD2} can serve as an instance of an advantageous quantum communication protocol.}. In this way, quantum protocols siphon the seemingly inaccessible ``excess ontological distinctness" of quantum preparations to an advantage in communication tasks. Communication tasks with similar channel constraints have been considered in \cite{armin2019}.
\begin{figure}     
    \centering
    \includegraphics[width=0.4\textwidth]{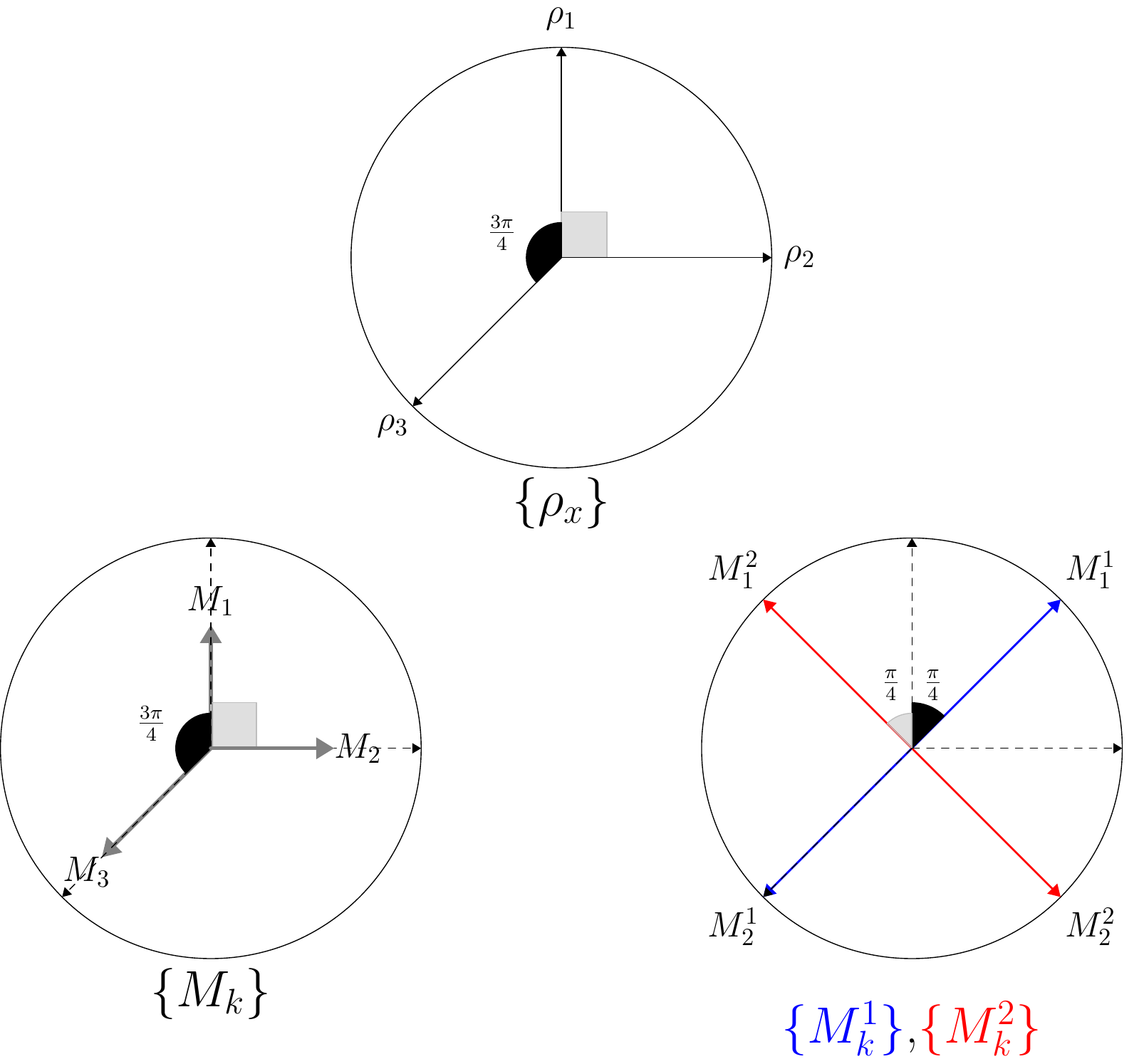}
    \caption{\label{PrepMeas4} This figure depicts the Bloch vectors in the x-y plane (x axis being vertical and y axis horizontal) of the Bloch sphere corresponding to the set of states $\{\rho_x=\frac{\mathbb{I}+\vec{n}_x\cdot \vec{\sigma}}{2}\}_{x\in\{1,2,3\}}$ (top) where $\vec{n_1}=[1,0,0]^T$, $\vec{n_2}=[0,1,0]^T$ and $\vec{n_3}=\frac{-\vec{n_1}-\vec{n_2}}{\sqrt{2}}$, the POVM $\{M_k\}_{k\in \{1,2,3\}}$ (bottom left) which optimally distinguishes these states such that $s_{Q}=\frac{2}{3}$, and the projective measurements $\{M^{1}_k\}_{k\in\{1,2\}},\{M^{2}_k\}_{k\in\{1,2\}}$ (bottom right) which violate Proposition \ref{propositionBOD2} such that $succ_{Q} = 3+\sqrt{2} \approx 4.414$. Note that the POVM elements $M_x= \alpha_x \left( \frac{\mathbb{I}+\vec{n}_x\cdot  \vec{\sigma}}{2} \right)$ where $\alpha_1 = \alpha_2 = \frac{\sqrt{2}}{1+\sqrt{2}}  , \alpha_3 = \frac{2}{1+\sqrt{2}}$ (bottom left) do not correspond to a vector in the Bloch sphere, instead we have employed a common heuristic depiction scheme entailing short thick gray arrows to denote the POVM elements, emphasize their directions and normalization \cite{kurzynski2006graphical}.}
\end{figure}



\subsection{Oblivious communication with bounded leakage}
Now we consider a cryptographically significant class of communication tasks referred to as oblivious communication tasks \cite{spekkens2009preparation,shp,PhysRevA.100.022108}. In such tasks, the sender Alice has a certain function of her inputs $f(x)\in \{1,\ldots,d_f\}$ (the oblivious function), the value of which is to be kept secret \footnote{Note that this is in contrast with the widely studied oblivious transfer tasks wherein the sender transfers one of potentially many pieces of information to a receiver, but remains oblivious as to what piece (if any) has been transferred.}. Here, we invoke a generalization of oblivious communication tasks wherein we tolerate certain amount of information leakage about the oblivious function $f(x)$. Specifically, the probability of distinguishing the preparations corresponding to different values of $f(x)$ out of a uniform ensemble is constrained such that,
\beq \label{opCond2} \nonumber
&s^{f(x)}_{\mathcal{O}} & :=  s^{f(x)=1,\ldots,f(x)=d_f}_{\mathcal{O}} \\ \nonumber
& &= {\scriptstyle \frac{1}{d_f}\max_{M}\Bigg\{\sum_{f(x)}\sum_x p(x|f(x))p(k=f(x)|P_x,M)\Bigg\}} \\
& & \leqslant \frac{1}{d_f}+\epsilon
\eeq

where $\epsilon\in (0,\frac{d_f-1}{d_f}]$ is the leakage parameter and, 
\be
p(x|f(x))=\frac{p(x)}{\sum_{x'|f(x')=f(x)}p(x')}.
\ee
This operational condition translates to the following ontological constraint in light of $BOD_P$ along with convexity of epistemic states,
\be \label{ontoCond2}
s^{f(x)}_{\Lambda}=\frac{1}{d_f}\int_{\Lambda}\max_{{f(x)}}\Bigg\{\sum_x p(x|f(x))\mu(\lambda|P_x)\Bigg\}d\lambda \leqslant \frac{1}{d_f}+\epsilon.
\ee
Observe that, the distinguishability inequality featured in Proposition \ref{propostionRhoEpistemic} can be readily posed as an instance of such a communication task. Specifically, consider an oblivious communication task wherein Alice receives a uniformly distributed pair of bits as her input $x=(a_1,a_2)$ where $a_1,a_2\in \{0,1\}$, $\forall \ a_1,a_2: \ p(a_1,a_2)=\frac{1}{4}$ and the oblivious function is simply the parity of these bits $f(x)=a_1\oplus a_2+1$. Bob on the other hand has to guess one of the Alice's input bits $z=a_y$ based on his own input bit $y\in\{1,2\}$. The coefficients of the success metric for this task are,
\be \label{ParityOblivious}
c(x,y,z)=
\begin{dcases}
    \frac{1}{8}, & \text{if } z=a_y, \\
    0, & \text{otherwise.}
\end{dcases}
\ee
This task is based on the well-known parity oblivious multiplexing task \cite{spekkens2009preparation}. Upon the relabelling, $P_1\equiv P_{a_1=0,a_2=0},P_2\equiv P_{a_1=1,a_2=1},P_3\equiv P_{a_1=0,a_2=1},P_4\equiv P_{a_1=1,a_2=0}$ along with $p=\frac{1}{2}+\epsilon$, it is easy to see that the inequality \eqref{rhoEpistemic} is equivalent to a bound on the success metric with coefficients \eqref{ParityOblivious}, i.e,. for $\epsilon$ leakage, and operational theories which satisfy $BOD_P$ along with convexity of epistemic states the success metric in this task is upper bounded as $succ_\mathcal{O} \leqslant \frac{3}{4} + \frac{\epsilon}{2}$.  \\
Moving on, the operational condition \eqref{opCond2} translates to the following constraint on $d$-levelled classical communication,
\be \label{ClassCond2}
s^{f(x)}_{C}=\frac{1}{d_f}\sum_{\omega}\max_{{f(x)}}\Bigg\{\sum_x p(x|f(x))p_{\mathcal{E}}(\omega|x)\Bigg\} \leqslant \frac{1}{d_f}+\epsilon.
\ee
Yet again, because the reasons stated above, for general classical communication models and in particular for $d$-levelled classical communication, the operational conditions \eqref{opCond2} and \eqref{ClassCond2} are equivalent to the ontological constrain \eqref{ontoCond2} \footnote{Observe that, in this case we did not make convexity of epistemic states an explicit requirement as the ontological model for general classical communication models (Definition \ref{Classicality}) adheres to convexity of epistemic states by default.}.  This in-turn ensures that the operational restrictions which hold for operational theories which satisfy $BOD_P$ and convexity of epistemic states also hold for classical communication models, specifically, for the aforementioned task specified by \eqref{ParityOblivious} $succ_C = succ_\Lambda \leqslant \frac{3}{4} + \frac{\epsilon}{2} $. Finally, employing appropriate relabelling, one can obtain an advantageous quantum prepare and measure protocols from the instances of violation of \eqref{rhoEpistemic} provided below Proposition \ref{propostionRhoEpistemic} \footnote{See related upcoming article(s) on the extent of preparation contextuality in imperfect parity oblivious communication task, by M. S. Leifer \textit{et al.} \cite{freda2018bounds}}. \\

\section{The implicate quantumness}
\label{implicateQuantumness}
In this section, we show how excess ontological distinctness is implicate in the explicate quantum departure from the well-known ontological notions of classicality. We go about this task by demonstrating that bounded ontological distinctness implies the other ontological notions of classicality as special cases so that the violation of the later implies excess ontological distinctness. \\
\textit{Generalized contextuality:}
First, we consider the notion of preparation noncontextuality. Based directly on the Leibniz's principle of ``identity of indiscernibles", preparation noncontextuality requires the epistemic states underlying two operationally equivalent preparations $P_1\equiv P_2$ to be identical, i.e. if $\forall \ k,M: p(k|P_1,M)=p(k|P_2,M)$ then preparation noncontextuality implies $\forall \ \lambda: \ \mu(\lambda|P_1)=\mu(\lambda|P_2)$. Clearly, preparation noncontextuality is a special case of $BOD_P$, i.e. whenever two given preparations are completely indistinguishable $s_\mathcal{O}=\frac{1}{2}$ then $BOD_P$ implies $\int_\Lambda\max\{\mu(\lambda|P_1),\mu(\lambda|P_2)\}d\lambda=1$, which only holds whenever $\forall \ \lambda: \ \mu(\lambda|P_1)=\mu(\lambda|P_2)$. Consequently, we have the implication,
\be \label{p-imp}
BOD_P \implies \text{preparation noncontextuality}.
\ee

Next, we consider the generalized notion of measurement noncontextuality which requires the response functions underlying two operationally equivalent measurement effects $[k|M^1] \equiv [k|M^2]$ to be identical, i.e. if $\forall \ P: \ p(k_1|P,M^1)=p(k_2|P,M^2)$ then measurement noncontextuality implies $\forall \ \lambda: \ \xi(k_1|\lambda,M^1)=\xi(k_2|\lambda,M^2)$. Consider two binary outcome measurements $\tilde{M^1},\tilde{M^2}$ which are basically coarse-grained versions of $M^1,M^2$, such that we retain the outcomes $k_1,k_2$ as their first outcomes, respectively, and relabel all other outcomes as their second outcomes. 
Now, we invoke the following natural property of ontological models,
\begin{Definition}
\textit{Coarse-graining of response functions:} If a measurement outcome $k$ is obtained by coarse-graining of two outcomes $k_1,k_2$, the response function corresponding to outcome $k$ is also coarse-graining of two response functions corresponding to outcomes $k_1,k_2$, i.e. $\forall \ \lambda: \ \xi(k|\lambda,\tilde{M})=\xi(k_1|\lambda,M)+\xi(k_2|\lambda,M)$, where $\tilde{M}$ is the coarse-grained version of $M$.
\end{Definition}
Therefore, from the definitions of $\tilde{M^1},\tilde{M^2}$ we have,
\beq \label{1k12k2}
& \forall \ \lambda: & \xi(1|\lambda,\tilde{M^1})=\xi(k_1|\lambda,{M^1}), \nonumber \\  
&& \xi(1|\lambda,\tilde{M^2})=\xi(k_2|\lambda,{M^2}).
\eeq  
The operational condition that $\forall \ P: \ p(k_1|P,M^1)=p(k_2|P,M^2)$ implies that the new measurements $\{\tilde{M^1},\tilde{M^2}\}$ are completely indistinguishable, i.e. for these measurements $m_\mathcal{O}=\frac{1}{2}$. In light of this observation and the consequent indistinguishability condition, $BOD_M$ implies, 
\beq \nonumber \label{inter}
 m_{\mathcal{O}} & = &   \frac{1}{2}\bigg(\max_\lambda\bigg\{\max\{\xi(1|\lambda,\tilde{M^1}),\xi(1|\lambda,\tilde{M^2})\} \\ \nonumber 
& & +\max\{1-\xi(1|\lambda,\tilde{M^1}),1-\xi(1|\lambda,\tilde{M^2})\}\bigg\}\bigg) \\
& = & \frac{1}{2}.
\eeq 
Using the fact that $\forall \ (a\geqslant0,b\geqslant0): \ \max\{a,b\}+\max\{1-a,1-b\} = 1+a+b-2\min\{a,b\}$, \eqref{inter} yields, 
\beq
&\forall \ \lambda: \ &  \xi(1|\lambda,\tilde{M^1})+\xi(1|\lambda,\tilde{M^2})  \nonumber \\
& & = 2\min\{\xi(1|\lambda,\tilde{M^1}),\xi(1|\lambda,\tilde{M^2})\} .
\eeq
Observe that this condition only holds when $\forall \ \lambda: \ \xi(1|\lambda,\tilde{M}^1)=\xi(1|\lambda,\tilde{M}^2)$, and subsequently \eqref{1k12k2} yields the following implication,
\be \label{m-imp}
BOD_M\implies \text{measurement noncontextuality}.
\ee
Finally, Spekkens' notion of transformation contextuality requires the transition schemes underlying a pair of operationally equivalent transformations $T_1\equiv T_2$ to be identical, i.e.
that if $\forall \ P,k,M:\ p(k|P,T_1,M)=p(k|P,T_2,M)$ then $\forall \ \lambda',\lambda:\ \gamma(\lambda'|\lambda,T_1)=\gamma(\lambda'|\lambda,T_2)$. Clearly, transformation noncontextuality is a special case of $BOD_T$, i.e. whenever two given transformations are completely indistinguishable $t_\mathcal{O}=\frac{1}{2}$ then $BOD_T$ implies $\max_{\lambda}\{ \int_{\Lambda}\max \{\gamma(\lambda'|\lambda,T_1),\gamma(\lambda'|\lambda,T_2)\}d\lambda' \}=1$, which only holds whenever $\forall \ \lambda',\lambda:\ \gamma(\lambda'|\lambda,T_1)=\gamma(\lambda'|\lambda,T_2)$. Consequently, we have the implication,
\be \label{t-imp}
BOD_T \implies \text{transformation noncontextuality}.
\ee
In general, universal generalized noncontextuality requires operationally equivalent pairs of preparations, measurements and transformations to have identical ontological counterparts, i.e. universal noncontextuality $\equiv$ (preparation noncontextuality $\wedge$ measurement noncontextuality $\wedge$ transformation noncontextuality). Correspondingly, we can define a notion of universal bounded ontological distinctness as,
\begin{Definition}
\textit{Bounded ontological distinctness ($BOD$):} The sets of epistemic states $\mathcal{P}_\Lambda$, response schemes $\mathcal{M}_\Lambda$ and transition schemes $\mathcal{T}_\Lambda$ underlying sets of $s$-distinguishable preparations $\mathcal{P}$, $m$-distinguishable measurements $\mathcal{M}$ and $t$-distinguishable transformations $\mathcal{T}$, must be $s$-distinct, $m$-distinct and $t$-distinct, respectively, i.e. 
\be \nonumber
BOD \equiv (BOD_P \wedge BOD_M \wedge BOD_T).
\ee
\end{Definition}
When formulated in this way, $BOD$ serves as a criterion for characterization of ontological models of a given operational theory. It is useful, at this point, to define a criterion for characterization of operational theories, namely, an operational theory or a fragment thereof is said to satisfy $BOD$, if there exists an ontological model which satisfies $BOD$ for all sets of prescribed operational preparations, measurements and transformations. Conversely, an operational theory or a fragment thereof is said to violate $BOD$ if there exists no ontological model which satisfies $BOD$ for all sets of prescribed operational preparations, measurements and transformations. Consequently, the implications \eqref{p-imp}, \eqref{m-imp} and \eqref{t-imp} yields the following combined implication, 
\be \label{imp}
BOD \implies \text{universal noncontextuality}.
\ee
While the implications \eqref{p-imp}, \eqref{m-imp}, \eqref{t-imp} and \eqref{imp} follow from the very definitions of generalized noncontextuality and $BOD$, in the following observation we address the validity of the reverse implications,
\begin{Observation} \label{impObservation}
The implications \eqref{imp} along with \eqref{p-imp} and \eqref{t-imp} are strictly unidirectional. 
\end{Observation}
\begin{proof}
In order to proof this thesis, we demonstrate (by explicit construction) that there exists fragments of quantum theory which admit a universally noncontextual ontological model but violate $BOD$. Consider a fragment of operational quantum theory entailing:
\begin{enumerate}
    \item any three distinct pure preparations $\{ P_i\}^3_{i=1}$ corresponding to pure quantum states $\{ \ket{\psi_i}\}^3_{i=1}$ along with their convex mixtures,
    \item three extremal transformations $\{T_j\}^3_{j=1}$ such that $T_{j}(P_{i})=P_{(i+j)\mod{3}}$ \footnote{here $(.)\mod{3}$ is defined such that $(3)\mod{3}=3$.}, along with their convex mixtures and,
    \item all possible quantum measurements $M\equiv \{M_k\}$.
\end{enumerate}
For this fragment we construct the following $\psi$-complete ontological model entailing,
\begin{enumerate}
    \item three ontic states $\{ \lambda_i\}^3_{i=1}$ underlying the preparations $\{ P_i\}^3_{i=1}$ corresponding to pure quantum states $\{ \ket{\psi_i}\}^3_{i=1}$,  
    \item three deterministic transition schemes $\{\{\gamma(\lambda'|\lambda,T_j)\}\}^3_{j=1}$ underlying the extremal transformations $\{T_j\}^3_{j=1}$
    such that $\forall \ i,j\in \{1,2,3\}: \ \gamma(\lambda_{(i+j)\mod{3}}|\lambda_i,T_j)=1$ and,
    \item response schemes $\{\xi(k|\lambda,M)\}$ underlying the quantum measurements $M\equiv \{M_k\}$ such that $\forall \ k,M,i\in\{1,2,3\}:  \xi(k|\lambda_i,M)=\Tr(\ketbra{\psi_i}{\psi_i}M_k)$.
\end{enumerate}
It is straightforward to verify that this model can reproduce all predictions of the fragment of operational quantum theory under consideration. Moreover, this ontological model could be considered a fragment of the Beltrametti-Bugajski ontological model of quantum theory. Consequently, following the line of reasoning in \cite{spekkens2005contextuality}, this ontological model is measurement noncontextual. \\
Moving on, to demonstrate that this ontological model is preparation noncontextual we employ the simple fact that the every point in a convex polytope has a unique decomposition in-terms of the extremal points when there are no more than three extremal points. Consequently, the density matrices corresponding to any two distinct convex mixtures of three pure preparations will be necessarily distinct. Finally, for any two distinct density matrices, there always exists quantum measurement effects with distinct observed statistics. Consequently, any $\psi$-complete ontological model for any three pure preparations is preparation noncontextual. \\
More formally, for the prepare and measure fragment of quantum theory described above, consider two operationally equivalent mixed preparations $P_a\equiv P_b$ such that $P_a\equiv \sum^3_{i=1} a_iP_{i}$ and $P_b\equiv \sum^3_{i=1} b_iP_{i}$ where $\sum^3_{i=1} a_i=\sum^3_{i=1} b_i=1$. Since these preparations are assumed to be operationally equivalent, we have $\forall \ k,M: \ p(k|P_a,M)=p(k|P_b,M)$. In particular, for the measurement effect $M_1\equiv \ketbra{\psi_1}{\psi_1}$ the operational equivalence $P_a\equiv P_b$ implies,
\beq \label{ob4c0}
a_1 + a_2\alpha_{1,2} + a_3\alpha_{1,3} = b_1 + b_2\alpha_{1,2} + b_3\alpha_{1,3},
\eeq 
where $\alpha_{i,j}=|\braket{\psi_i|\psi_j}|^2$. Using the fact that $\sum_ia_i=\sum_ib_i=1$ \eqref{ob4c0} yields,
\be \label{ob4c1}
(a_1-b_1)(1-\alpha_{1,3})= (b_2-a_2)(\alpha_{1,2}-\alpha_{1,3}) .
\ee
Similarly for the measurement effects $M_2\equiv \ketbra{\psi_2}{\psi_2}$ and $M_3\equiv \ketbra{\psi_3}{\psi_3}$ the operational equivalence $P_a\equiv P_b$ yields,
\beq \label{ob4c2} \nonumber
&(a_1-b_1)(\alpha_{1,2}-\alpha_{2,3}) & = (b_2-a_2)(1-\alpha_{2,3}), \\ 
&(a_1-b_1)(1-\alpha_{1,3})& = (b_2-a_2)(1-\alpha_{2,3} ).
\eeq
Since, the three pure preparations under consideration correspond to distinct pure states $\forall \ i\in\{1,2,3\}, j \in \{1,2,3\} \setminus \{i\}: \ \alpha_{i,j}<1$, the conditions \eqref{ob4c1} and \eqref{ob4c2} can only be simultaneously satisfied if $\forall \ i\in\{1,2,3\}: \ a_i=b_i$.
This implies the uniqueness of the decomposition of mixed preparations in terms of the three pure preparations. As there is no scope for any non-trivial equivalence condition, the $\psi$-complete ontological model described above is preparation noncontextual. Moreover, as the consequences of bounded ontological distinctness, namely, Proposition \ref{propositionBOD} and Proposition \ref{propositionBOD2} are violated employing just three pure quantum preparations (see Subsection \ref{markSection1}, FIG. \ref{PrepMeas1}, FIG. \ref{qubitQutritRandom}, and end of Subsection \ref{markmeup2}, FIG. \ref{PrepMeas4}, respectively), this completes the proof for strict unidirectionality of the implication \eqref{p-imp}. \\
Now, consider a pair of operationally equivalent mixed transformations $T_a\equiv T_b$ such that $T_a\equiv \sum^3_{j=1} a_jT_{j}$ and $T_b\equiv \sum^3_{j=1} b_iT_{i}$ where $\sum^3_{i=1} a_i=\sum^3_{i=1} b_i=1$. The effect of these transformations on a pure preparation $P_i$, i.e. $T_a(P_i)=\sum^3_{j=1} a_jP_{(i+j)\mod{3}}$ and $T_b(P_i)=\sum^3_{j=1} b_jP_{(i+j)\mod{3}}$. As the transformations are equivalent they yield equivalent post transformation preparations $\sum^3_{j=1}a_jP_{(i+j)\mod{3}}\equiv\sum^3_{j=1}b_jP_{(i+j)\mod{3}}$, which because of the uniqueness of decomposition proved above, can only hold when $\forall \ j\in\{1,2,3\}: \ a_j=b_j$. This implies uniqueness of the decomposition of mixed transformations in terms of the three extremal transformations. As there is no scope for any non-trivial equivalence condition, the ontological model described above is transformation noncontextual. Moreover, as Proposition \ref{propositionBODT} is violated employing just three unitary quantum transformations (see Subsection \ref{markme3}) which satisfy the aforementioned operational requirement, this completes the proof for strict unidirectionality of the implication \eqref{t-imp}. Summarizing, instances of the operational
fragment of quantum theory described above violate $BOD_P$ and $BOD_T$, but the provided $\psi$-complete ontological model is preparation noncontextual, transformation noncontextual, and measurement noncontextual, thereby completing the proof for strict unidirectionality of the combined implication \eqref{imp}. \\
\end{proof}
As neither measurement noncontextuality nor bounded ontological distinctness for measurements can be violated by quantum theory on their own, the unidirectionaility of the implication \eqref{m-imp} is uncertain.

\textit{Kochen-Specker contextuality and Bell nonlocality:}
Notice that the aforementioned implications follow from the very definitions of the ontological principles under consideration, without invoking any other formalism dependent assumptions. With the aid of certain quantum formalism dependent assumptions $QT$, including self duality of states and measurements effects, along with the corresponding ontological constraints, we have the implication \cite{leifer2013maximally,leifer2014quantum,leifervideo}: $\text{preparation noncontextuality} \overset{QT}{\implies} \text{Kochen-Specker noncontextuality}$. This implication is valid for all ontological models that reproduce the predictions of quantum theory.

Moreover, Bell's local-causality is a special case of Kochen-Specker noncontextuality. Specifically, spatially separated measurements employed in set-ups for Bell's local-causality imply commutation relations employed in the definition of contexts in Kochen-Specker contextuality. In general, for non-signaling operational theories we have the implication:
$\text{Kochen-Specker noncontextuality} \overset{NS}{\implies} \text{Bell's local-causality}$. Consequently, this leads to the following combined implication,

\beq \nonumber
& BOD_P & \implies \text{preparation noncontextuality} \\ \nonumber
& & \overset{QT}{\implies} \text{Kochen-Specker noncontextuality} \\
& &\overset{NS}{\implies} \text{Bell's local-causality}.
\eeq
Therefore, any quantum violation of Bell's local-causality, Kochen-Specker noncontextuality or preparation noncontextuality implies violation of $BOD_P$ or excess ontological distinctness of quantum preparations.

\section{Concluding remarks}
The ontological framework features certain philosophically motivated principles. These ontological principles provide certain exclusively operational phenomena an ontological basis in the form of ontological constraints. The subsequent quantum violation of the consequences of such principles discards the ontological models that adhere to these constraints as plausible ontological explanations of the operational theory under consideration.

\textit{From complete indistinguishability to distinguishability:} Remarkably, all well-known ontological principles employ operational conditions pertaining specifically to the indistinguishability of associated physical entities. For instance, Bell's local causality provides an ontological basis for operational non-signaling correlations. The non-signaling condition, an operational pre-requisite of Bell's local causality, is an indistinguishability condition, specifically, the condition requires each spatially separated party to not be able to distinguish between the different measurements employed by the other parties. Similarly, Kochen-Specker noncontextuality is accompanied by an operational non-disturbance condition for particular sets of measurements, which is again an operational indistinguishability condition for measurements. Specifically, for such a set of measurements, the condition requires one to not be able to distinguish between the rest of the measurements on the basis of the outcome statistics of any measurement. Moreover, Spekkens' noncontextuality relies directly on an operational equivalence of preparations, measurements, and transformations which can clearly be interpreted as an operational indistinguishability condition of respective physical entities. Yet another ontological principle of non-retrocausal time-symmetric ontology \cite{leifer2017time} requires each temporally separated party to not be able to distinguish between the different (measurement or preparation) settings of the other parties. The subject of this work, the principle of bounded ontological distinctness provides an ontological basis to the maximal distinguishability arbitrary sets of preparations, measurements, and transformations, in the distinctness of their ontological counterparts. 

\textit{Measure dependence:} The principle of bounded ontological distinctness inherits the ambiguity with regards to measures of operational distinguishability and ontological distinctness from its philosophical premise of the natural generalization of the Leibnitz principle. The ontological notion itself does not impose a particular measure, and any operational measure of distinguishability can be employed as long as its ontological counterpart is properly defined as a measure of distinctness. The particular measure we have employed is the probability of minimum error discrimination, which is a natural choice for a measure of distinguishability. Furthermore, it is based on the maximum success probability of an operational task, namely, minimum error discrimination, and consequently, it is independent of any particular theory. When one employs operational prescriptions pertaining to a given operational theory as resources in the minimum error discrimination task, the maximum success probability corresponds to an operational measure of distinguishability. On the other hand, if one employs the ontological counterparts as resources along with complete access to (and fine-grained control over) the ontic-state of the system, the maximum success probability of minimum error discrimination corresponds to a measure of ontological distinctness.

\textit{Violation without auxiliary assumptions:} All of the aforementioned ontological principles including bounded ontological distinctness obtain their primary ontological constraints by requiring ontological models to adhere to corresponding operational conditions on the ontological level i.e. conditioned on the ontic state of the physical system. However, the primary ontological constraints of the well-known ontological principles are (on their own) not enough to warrant consequences directly in contradiction with the predictions of quantum theory. Therefore, they invoke certain auxiliary constraints to enable a quantum violation. For instance, the auxiliary assumption of outcome determinism is quintessential to the demonstration of Kochen-Specker noncontextuality, while Spekkens' preparation and transformation noncontextuality require the auxiliary ontological property of convexity of epistemic states and transition schemes respectively to enable their quantum violations. Finally, the quantum demonstration of Spekkens' measurement contextuality requires either outcome determinism or Spekkens' preparation noncontextuality as an additional assumption. The inclusion of these additional assumptions leads to dilution of the implication of corresponding quantum violations as they could be attributed to the violation of the auxiliary assumptions leaving the primary assumptions intact. Note that, just like measurement noncontextuality, bounded ontological distinctness for measurements requires either outcome determinism or bounded ontological distinctness for preparations to enable a quantum violation. However, unlike these ontological principles, the primary ontological assumption of bounded ontological distinctness for preparations and transformations is enough to warrant a quantum violation on its own. This, in turn, leads to an unambiguous ontological implication, that quantum preparations, and transformations are more ontologically distinct than they are operationally distinguishable. 

\textit{Non-zero measure operational perquisite:} 
One of the key issues with the experimental tests of the other ontological principles \cite{hensen2015loophole,kirchmair2009state,mazurek2016experimental},
lies in the associated operational conditions that the corresponding physical entities have to adhere to. Specifically, the well-known ontological notions of classicality, for their refutation, require corresponding operational conditions to hold. For instance, the refutation of Bell's local causality requires the experimental data to adhere to non-signaling, the refutation Kochen-Specker noncontextuality requires the experimental data to adhere to non-disturbance and for the refutation of Spekkens' generalized noncontextuality, an equivalence condition of mixtures of preparations (termed oblivious condition) must hold. As detailed above, these conditions require the zero measure indistinguishability of the associated physical entities, and consequently their experimental tests are susceptible to the persistent finite precision loophole.  However, the operational condition associated with bounded ontological distinctness, namely, distinguishability of physical entities has a considerably larger spectrum, and serves as an alternative approach to address the finite precision loophole in the experimental tests of other ontological principles.

\textit{The implicate quantumness:} Our perception or understanding of certain phenomena might differ, or might be characterized by, varying principal factors, depending on contexts such as scales. The implicate or ``enfolded'' is the deeper, more fundamental order associated with these phenomena while the explicate or ``unfolded'' include the abstractions or observations in light of specific contexts. 
\epigraph{``For instance, a circular table looks like an ellipse from various directions. But we know that those are the appearances of a single circular form. So we represent the table as a circle. We say, `that what's it is, a solid circle'"}{-- David Bohm \cite{bohm2004thought}.}
We show that bounded ontological distinctness, directly and indirectly, implies the other ontological principles, so that the \textit{quantum} violation of the latter implies the violation of the former \eqref{implicateQuantumness}. In particular, we show that maximal $\psi$-epistemicity \cite{PhysRevLett.112.250403} is a restricted case of bounded ontological distinctness of preparations applied to a pair of pure quantum preparations. Similarly, preparation, measurement and transformation noncontextuality emerges from bounded ontological distinctness of the respective physical entities. Consequently, universal noncontextuality is a direct implication of bounded ontological distinctness. Moreover, via Observation \ref{impObservation} we demonstrate that the set of generalized noncontextual ontological models strictly contains the set of ontological models that satisfy bounded ontological models, deeming these implications to be strictly unidirectional.  
While these implications follow by the very definitions of these ontological principles, we show that under certain quantum theory dependent ontological assumptions, bounded ontological distinctness for preparations implies Kochen-Specker noncontextuality and Bell's local causality. This provides a crucial unifying outlook, namely, the quantum violations of the other ontological principles are, in essence, a demonstration of the implicate quantum excess ontological distinctness, deeming it to be a fundamental feature of quantum ontology. Here, it is important to note that, without the quantum formalism dependent assumptions such as self-duality of states and measurement effects, the implication from bounded ontological disticntess for preparations via preparation noncontextuality to Kochen-Specker noncontextuality (and Bell's local causality) falls apart. This is explicated by the fact that the fragments considered in Observation \ref{KSnogoPair} and \ref{impObservation} are Kochen-Specker contextual for quantum systems of dimension three or more, however, they satisfy bounded ontological distinctness for preparations, and preparation noncontextuality, respectively. 

 \textit{Requirements for quantum violation:} We have demonstrated the violation of bounded ontological distinctness for preparations, and transformations, while employing three instances of pure two-dimensional quantum preparations and unitary transformations, respectively, along with three binary outcome measurements for Propositions \ref{propositionBOD} and \ref{propositionBODT}, and two binary outcome measurements for Proposition \ref{propositionBOD2} \footnote{Note that, this does not include the POVM that verifies the operational perquisite, and maximizes the distinguishability of the preparations under consideration, as it can be replaced by tomography or by suitable semi-device independent dimensional assumption.}. Moreover, Observation \ref{KSnogoPair} brings forth the fact any operational fragment of quantum theory entailing a pair of pure quantum preparations, and all possible measurements satisfies bounded ontological distinctness.
 Moreover, while restricting the operational distinguishability of a pair of mixed quantum preparations we demonstrate the violation of bounded ontological distinctness and convexity of epistemic states employing four preparations and two measurements (Proposition \ref{propostionRhoEpistemic}). Finally, we observed that bounded ontological distinctness for measurements cannot be violated on its own because of the existence of Beltrametti-Bugajski $\psi$-complete ontological model. However, invoking the additional assumption of outcome determinism, we show that all pairs of two-dimensional quantum projective measurements which are neither completely indistinguishable nor perfectly distinguishable violate bounded distinctness for measurements. 
 

\textit{Distinguishability and distinctness as intrinsic properties:} This work is based on an overarching perception, namely, the distinguishability, specifically, the maximal operational probability of distinguishing a set of preparations, measurements or transformations forms an intrinsic property of this set. This is so because the maximization involved relieves the distinguishing probability of a set of a particular type of physical entity of its dependence on other types of physical entities. Similarly, ontological distinctness or the maximal probability of distinguishing sets of epistemic states, response schemes or transition schemes form an intrinsic property of the set of ontological entities. in light of this insight, the consequences of bounded ontological distinctness, specifically the inequalities featured in the Propositions \ref{propositionBOD}, \ref{propostionRhoEpistemic}, \ref{propositionBODT} and \ref{propositionBOD2} present a novel perspective. Specifically, these propositions relate two distinct operational properties pertaining to the distinguishability of respective physical entities. For instance, Proposition \ref{propositionBOD} bounds the maximal average pair-wise distinguishability, based on maximal distinguishability for a set of three operational preparations. The consequent quantum violation of these propositions implies that the conflict with bounded ontological distinctness lies in the relation between these intrinsic properties of the set of quantum physical entities under consideration. 

\textit{Quantifying the extent of excess ontological distinctness:} The particular form of the Propositions \ref{propositionBOD}, \ref{propostionRhoEpistemic}, \ref{propositionBODT}, \ref{propositionBOD2}, and of the inequalities therein, allow us to infer a lower bound on the extent of excess ontological distinctness.  
Remarkably, the Kochen-Specker's maximally $\psi$-epistemic model saturates the lower bounds on the extent of quantum excess ontological distinctness (Observations \ref{KSpropositionBOD}, \ref{KSrhoepistemic}). This in-turn highlights the significance of maximally $\psi$-epistemic models and demonstrates that the inequalities under consideration are tight. Furthermore, more we employ the existence of this model to substantiate the fact any operational fragment of quantum theory entailing a pair of pure quantum preparations, and all possible measurements satisfies bounded ontological distinctness (Observation \ref{KSnogoPair}). Moreover, as inequality featured in Proposition \ref{propostionRhoEpistemic} is a generalization of the preparation noncontextual inequality based on parity oblivious multiplexing \cite{spekkens2009preparation}, and the states employed in Observation \ref{KSrhoepistemic} adhere to the operational equivalence condition, the Observation \ref{KSrhoepistemic} also yields a lower bound on the extent of preparation noncontextual, namely, the indistinguishable pair of mixed states employed therein must have $\frac{1}{\sqrt{2}}$-distinct epistemic states underlying them.  

\textit{Powering advantage in communication tasks:} The success in communication tasks is closely tied to the ability of the receiver to distinguish the sender's preparations. Specifically, the sender encodes her input data onto operational preparations and the success depends upon how the well the receiver can distinguish these preparations to figure out the sender's input \footnote{This is evident in the proof methodology of Proposition \ref{propositionBOD2} and the casting of Proposition \ref{propositionBOD} and \ref{propostionRhoEpistemic} as communication tasks}.  The excess ontological distinctness of the quantum preparations fuels the quantum advantage in communication tasks presented in this work wherein we constrain the distinguishability of the sender's operational preparations \cite{PhysRevA.100.022108}. Specifically, the provided quantum communication protocols siphon the excess ontological distinctness of the associated epistemic states to an advantage over classical unbounded communication protocols and theories that have ontological models that adhere to bounded ontological distinctness. Moreover, as the advantage in oblivious communication tasks witnesses preparation contextuality, the oblivious communication tasks with bounded leakage which are powered by excess ontological distinctness of the sender's preparations demonstrate that the non-zero measure operational condition accompanying bounded ontological distinctness may be ported to other ontological principles.  

\textit{Avenues for future investigation:} The underlying assumption in our implications is the existence of an ontology adhering to the standard ontological framework with epistemic states, response schemes and transition schemes modeling operational preparations, measurements, and transformations. Then one of the implications of the violation of the ontological principles could be non-existence on an ontology, hinting towards a purely perspectival standpoint \cite{leifer2014quantum,frauchiger2018quantum,PhysRevA.94.052127}. More interestingly, to preserve bounded ontological distinctness of quantum entities, the prescriptions of the standard ontological framework might be altered, giving way to an exotic quantum ontology. 
It might be worthwhile to investigate the features of operational theories that enable excess ontological distinctness. Even though it is a fact that the quantum violation of other ontological principles implies excess ontological distinctness, it will be interesting to investigate how certain operational features propagate excess ontological distinctness to the violation of the other well-known ontological principles in the respective scenarios. 
Moreover, in this work we recovered Kochen-Specker noncontextuality from bounded ontological distinctness for preparations via preparation noncontextuality, which necessities several quantum formalism dependent assumptions including self-duality of states and measurement effects, alternatively, it might be worthwhile to try and recover Kochen-Specker noncontextuality from bounded ontological distinctness of measurements along with outcome determinism which is intrinsic to the former.  As the recently introduced framework of ``operational fine turnings" \cite{catani2020mathematical}  succinctly incorporates the well-known ontological notions of classicality employing the umbrella of a Leibniz-type principle termed ``no fine tuning", our results might provide a way to generalize the framework in so far as to include ``no fine tuning" of non-zero measure operational conditions. It would be interesting to find out whether the non-zero measure operational conditions accompanying bounded ontological distinctness could lead to noise robust experimental demonstrations. 

On the technical front, in light of the measure dependent formulation of bounded ontological distinctness, it will be instructive to study the implications of formulations based on other measures of distinguishability such as the maximum success probability of unambiguous discrimination and compare the consequences. In particular, as bounded ontological distinctness formulated with the maximum probability of minimum error discrimination is a generalization of maximal $\psi$-epistemicity based on the symmetric overlap of epistemic states \cite{PhysRevLett.112.250403}, it would be interesting to see how bounded ontological distinctness relates to the other definition of maximal $\psi$-epistemicity based on the asymmetric overlap of epistemic states \cite{leifer2013maximally}.
It might be interesting to study the implications of the generalizations of Propositions \ref{propositionBOD},\ref{propositionBODT} and \ref{propositionBOD2} to arbitrary number of preparations and transformation, respectively. It will be interesting to find out if bounded ontological distinctness of a pair of mixed preparation along with convexity of epistemic states can be violated whilst employing just three instances of two-dimensional quantum preparations. As neither measurement noncontextuality nor bounded ontological distinctness for measurements can be violated on their own, it would be interesting to demonstrate the unidirectionality of the implication \eqref{m-imp} employing outcome determinism.
While we have demonstrated quantum advantage in communication tasks with constraints on the distinguishability of the sender's preparations, we claim that quantum advantage may be retrieved in multiparty communication tasks while restricting the distinguishability of measurements or intermediate transformations. This should be formalized into a semi-device independent framework, equipped with key distribution and randomness certification protocols, based on theory independent departure from the predictions of classical theories. Finally, physical principles such as information causality \cite{pawlowski2009information} and macroscopic reality attempt to explain why quantum resources do not violate Bell inequalities to their algebraic maximum and enable insights into peculiarities of the quantum formalism and hierarchies of semi-definite programs aid in obtaining upper-bounds on quantum violations of these Bell inequalities \cite{navascues2008convergent}. It would be worthwhile to conjure such physical principles and semi-definite hierarchies that restrict the violation of bounded ontological distinctness inequalities to the quantum maximum, analytically and numerically respectively. Finally, it will be interesting to investigate the role of excess ontological distinctness in known instances of advantageous quantum information processing and computation protocols, and to come with new classes of information processing and computation tasks powered by excess ontological distinctness.


\textit{Conceptual insight:} For a set of physical entities, their distinguishability quantifies how well we can tell them apart employing a given physical theory. On the other hand, their distinctness quantifies how distinct these entities actually are (in reality). In classical theories, and in general classical thought, when it comes to sets of physical entities, distinguishability is synonymous to distinctness, i.e., \textit{what you see is what you get}. However, in the work, we demonstrate that this is not the case for quantum physical entities. Quantum theory posits sets of physical entities that must be more distinct than they are distinguishable. 

\section{Acknowledgements}
We thank M. Paw\l owski, J. H. Selby, M. Oszmaniec, M. Farkas, N. Miklin, M. S. Leifer and M. Banik for insightful discussions. We are grateful to R. W. Spekkens and M. S. Leifer for their lectures on quantum foundations at \textit{Perimeter Institute Recorded Seminar Archive}. We are grateful to the anonymous reviewers for their detailed constructive criticisms. We are grateful to the folks over at math.stackexchange for their prompt assistance \cite{3247130}. 
This research was supported by FNP grants First TEAM/2016-1/5, First TEAM/2017-4/31, NCN grant 2016/23/N/ST2/02817, and NCN grant SHENG 2018/30/Q/ST2/00625. The numerical optimization was carried out using \href{https://ncpol2sdpa.readthedocs.io/en/stable/index.html}{Ncpol2sdpa} \cite[]{wittek2015algorithm}, \href{https://yalmip.github.io/}{YALMIP} \cite[]{Lofberg2004}, \href{https://www.mosek.com/documentation/}{MOSEK} \cite[]{mosek} and \href{https://cvxopt.org/}{CVXOPT} \cite[]{andersen2013cvxopt}. AC acknowledges Tool for \href{https://www.youtube.com/watch?v=-_nQhGR0K8M}{this}.

\bibliographystyle{alphaarxiv}
\bibliography{ref}

\clearpage
\onecolumngrid
\section{*Auxiliary plots}

\setcounter{totalnumber}{4}

This section presents three auxiliary plots. The first two plots contain results of our numerical simulations showcasing extensive violation of \eqref{BODineq} and \eqref{rhoEpistemic}. These plots highlight the experimental robustness of distinguishability as an operational pre-requisite, as in these simulations the operational condition and the value of the inequality are decided after randomly sampling arbitrary triplets and quadruplets of two-dimensional quantum preparations respectively. The third plot serves to aid visualization of quantum violation \eqref{refMe} of Proposition \ref{rhoEpistemic}.
\begin{figure}[ht]     
    \centering
    \includegraphics[scale=0.45]{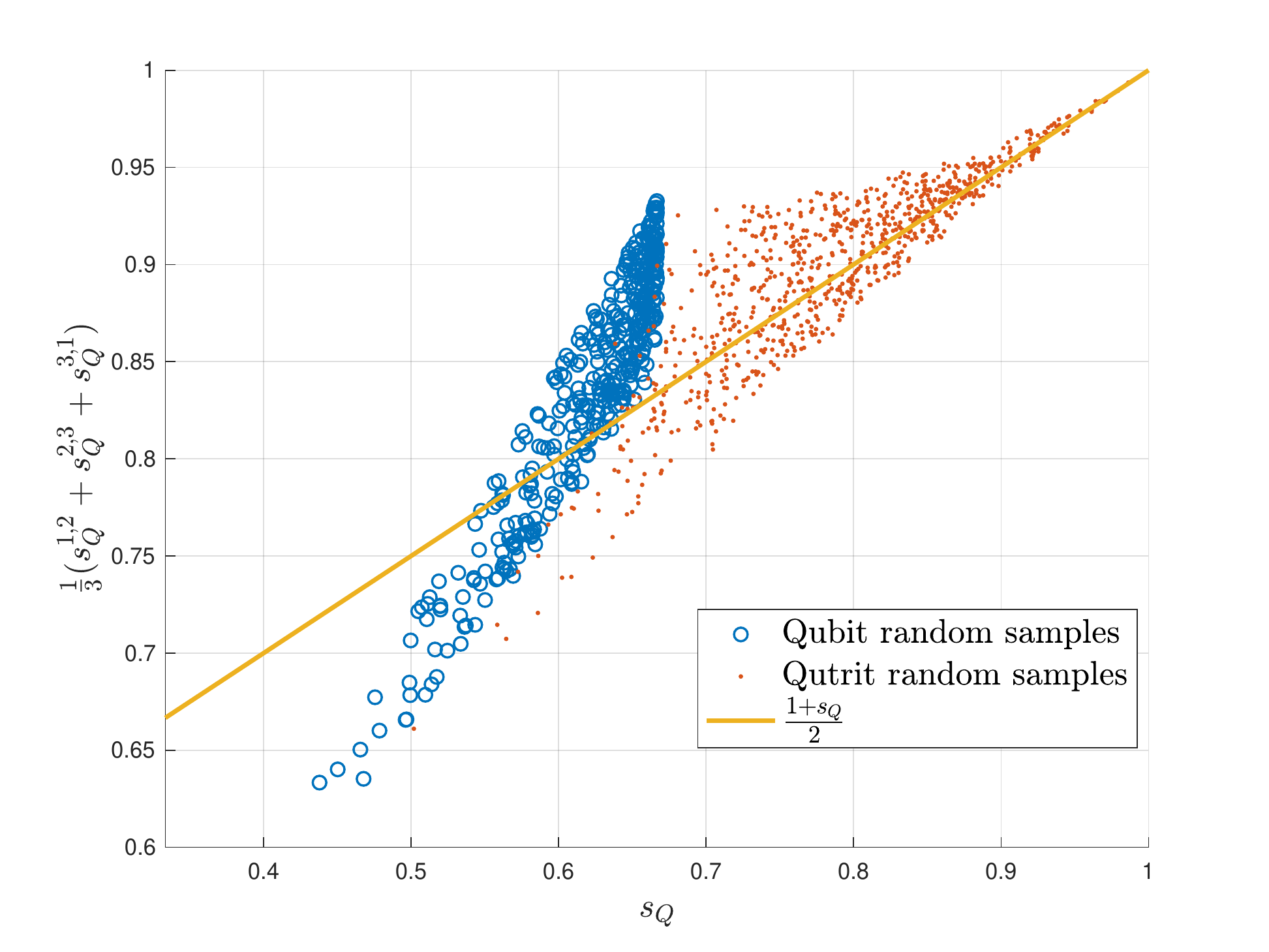}
    \caption{\label{qubitQutritRandom} A plot of maximal average pairwise distinguishability $\frac{1}{3}(s^{1,2}_{Q}+s^{2,3}_{Q}+s^{3,1}_{Q})$ vs. maximal distinguishability $s$ for three randomly picked pure qubits (blue circles) and three pure qutrits (orange dots), based on a uniform (according to Haar measure) distribution on the unit hypersphere. The yellow line represents the maximal average pairwise distinguishability \eqref{BODineq} when the underlying epistemic states adhere to bounded ontological distinctness. Over $72\%$ of randomly picked triplets of pure qubits and over $58\%$ of randomly picked triplets of pure qutrits exhibit excess ontological distinctness. }
\end{figure}
\begin{figure}[ht]
    \centering
    \includegraphics[scale=0.45]{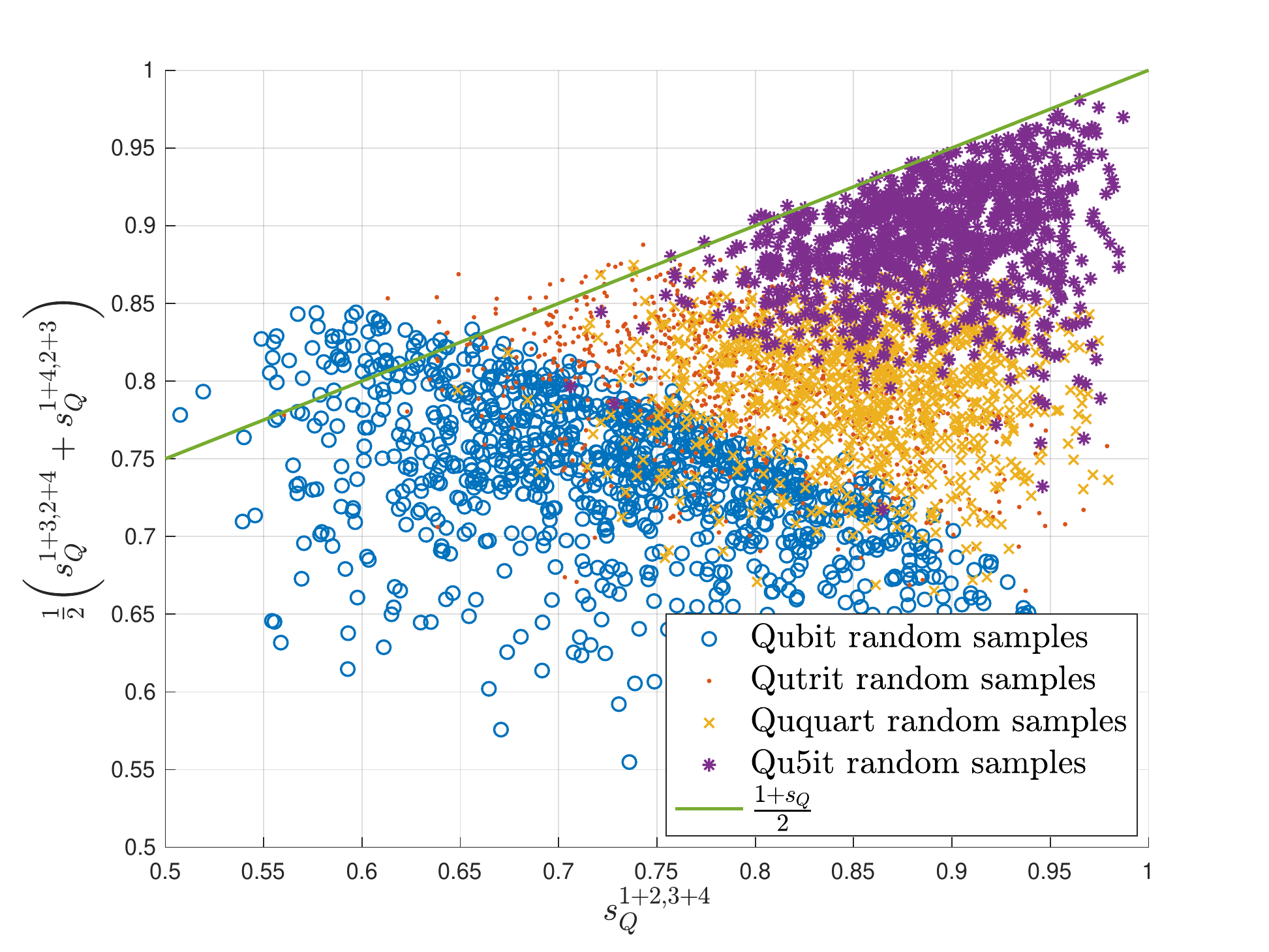}
    \caption{\label{qu2345itRandom} A plot of maximal average distinguishability of pairs of mixtures $\{\rho_{1+3},\rho_{2+4}\}$ and $\{\rho_{1+4},\rho_{2+3}\}$,  $\frac{1}{2}(s^{1+3,2+4}_{Q}+s^{1+4,2+3}_{Q})$ vs. maximal distinguishability of $\{\rho_{1+2},\rho_{3+4}\}$, $s^{1+2,3+4}_Q$ for four randomly picked pure qubits (blue circles), pure qutrits (orange dots), pure ququarts (yellow crosses) and pure qu$5$its (purple asterisks), based on a uniform (according to Haar measure) distribution on the unit hypersphere. The green line represents the maximal average distinguishability of $\{\rho_{1+3}, \rho_{2+4}\}$ and $\{\rho_{1+4}, \rho_{2+3}\}$ \eqref{rhoEpistemic} when the underlying epistemic states adhere to bounded ontological distinctness and convexity of epistemic states.}
\end{figure}
\begin{figure}[ht]
    \centering
    \includegraphics[scale=0.45]{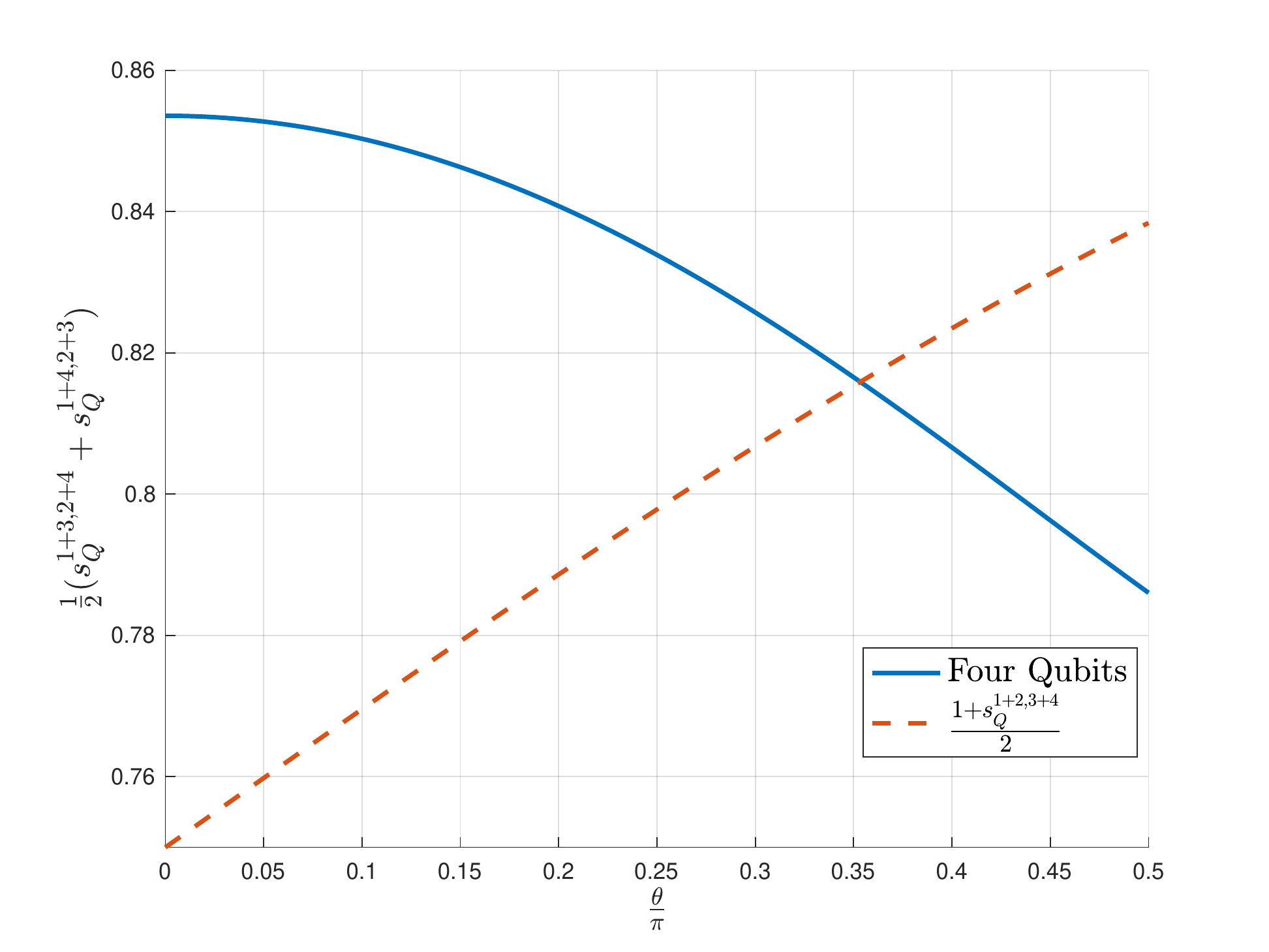}
    \caption{ \label{fourQubit}
    A plot of maximal average distinguishability of the pairs $\{\rho_{1+3},\rho_{2+4}\}$ and $\{\rho_{1+4},\rho_{2+3}\}$, $\frac{1}{2}(s^{1+3,2+4}_Q +s^{1+3,2+4}_Q)$ \eqref{refMe} (blue line) along with the bound $\frac{1+s^{1+2,3+4}_Q}{2}$ (orange dashed line) from Proposition \ref{propostionRhoEpistemic} against $\frac{\theta}{\pi}$. The qubits violate the inequality \eqref{rhoEpistemic} for the range $0\leq\frac{\theta}{\pi}\lessapprox\frac{1}{2\sqrt{2}}$.}
\end{figure}
\end{document}